\documentclass[11pt]{article}

\usepackage{soul}
\usepackage{url}
\usepackage[hidelinks]{hyperref}
\usepackage[utf8]{inputenc}
\usepackage[small]{caption}
\usepackage{graphicx}
\usepackage{booktabs}
\usepackage[T1]{fontenc}

\usepackage{amsmath,multirow}

\usepackage{amsthm,thm-restate,thmtools}
\usepackage{amssymb}

\usepackage{enumerate}
\usepackage{etoolbox}
\usepackage{subcaption}
\usepackage{xcolor,xfrac}
\usepackage[sort,numbers]{natbib}
\usepackage{hyperref}
\usepackage[english]{babel}
\usepackage{graphicx,charter}
\usepackage{framed}
\usepackage[top=1 in,bottom=1in, left=1 in, right=1 in]{geometry}

\usepackage{xcolor}
\usepackage{hyperref}
\hypersetup{
	colorlinks   = true,
	linkcolor    = red, %
	urlcolor     = teal, %
	citecolor    = blue %
}

\usepackage{algorithm}
\usepackage{algpseudocodex}
\usepackage{doi}

\usepackage{xspace}
\usepackage{amsthm}
\usepackage{cleveref}
\usepackage{todonotes}
\usepackage{comment}
\usepackage{fontawesome}
\usepackage{tikz}
\usetikzlibrary{backgrounds,shapes.geometric}
\usepackage{enumitem}

\newcommand{\N}{\mathbb{N}}

\newcommand{\probName}[1]{\textsc{#1}}

\newcommand{\agents}{N}
\newcommand{\numAgents}{n}

\newcommand{\orders}{V\setminus\{h\}}
\newcommand{\numOrders}{m}

\newcommand{\wFn}{\omega}

\newcommand{\alloc}{\pi}

\DeclareMathOperator{\parent}{parent}
\DeclareMathOperator{\leaves}{leaves}
\newcommand{\numLeaves}{L}

\DeclareMathOperator{\dist}{dist}

\DeclareMathOperator{\cost}{cost}

\DeclareMathOperator{\MMSshare}{MMS-share}

\DeclareMathOperator{\PoNW}{PoNW}

\newcommand{\NP}{\textsf{NP}\xspace}
\newcommand{\NPh}{\NP-hard\xspace}

\newcommand{\NPc}{\NP-complete\xspace}

\newcommand{\FPT}{\textsf{FPT}\xspace}
\newcommand{\XP}{\textsf{XP}\xspace}
\newcommand{\W}[1][1]{\textsf{W[#1]}\xspace}
\newcommand{\Wh}[1][1]{\W[#1]-hard\xspace}

\newcommand{\Oh}[1]{\ensuremath{{\mathcal{O}\left(#1\right)}}}

\newtheorem{definition}{Definition}
\newtheorem{theorem}{Theorem}
\newtheorem{lemma}{Lemma}

\newtheorem{proposition}{Proposition}
\newtheorem{corollary}{Corollary}
\newtheorem{claim}{Claim}
\crefname{claim}{Claim}{Claims}
\newenvironment{claimproof}[1]{\par\noindent\emph{Proof.}\hspace{0.15cm}#1}{\hfill~$\blacktriangleleft$\medskip}

\newcommand{\Yes}{\emph{yes}\xspace}
\newcommand{\No}{\emph{no}\xspace}

\title{The Algorithmic Landscape of Fair and Efficient Distribution of Delivery Orders in the Gig Economy}

\author{
	Hadi Hosseini \\ 
	Penn State University, USA\\ 
	\texttt{hadi@psu.edu}
	\and 
	Šimon Schierreich \\ 
	Penn State University, USA \\ 
	Czech Technical University, Czechia\\
	\texttt{schiesim@fit.cvut.cz}
}

\date{}

\begin{document}

\maketitle 

\begin{abstract}
	Distributing services, goods, and tasks in the gig economy heavily relies upon on-demand workers (aka agents), leading to new challenges varying from logistics optimization to the ethical treatment of gig workers. We focus on fair and efficient distribution of delivery tasks---placed on the vertices of a graph---among a fixed set of agents. We consider the fairness notion of minimax share (MMS), which aims to minimize the maximum (submodular) cost among agents and is particularly appealing in applications without monetary transfers. We propose a novel efficiency notion---namely non-wastefulness---that is desirable in a wide range of scenarios and, more importantly, does not suffer from computational barriers. Specifically, given a distribution of tasks, we can, in polynomial time, i) verify whether the distribution is non-wasteful and ii) turn it into an equivalent non-wasteful distribution. Moreover, we investigate several fixed-parameter tractable and polynomial-time algorithms and paint a complete picture of the (parameterized) complexity of finding fair and efficient distributions of tasks with respect to both the structure of the topology and natural restrictions of the input. Finally, we highlight how our findings shed light on computational aspects of other well-studied fairness notions, such as envy-freeness and its relaxations.
\end{abstract}

\section{Introduction}

Distributing services, goods, and tasks in today's economy increasingly relies upon on-demand gig workers. In particular, many e-commerce platforms and retail stores utilize freelance workers (in addition to their permanent employees) to distribute goods in an efficient manner. Naturally, this so-called `gig economy' involves many workers (aka agents), leading to new challenges from logistical and ethical perspectives. While the logistical aspect of this problem has been studied from an optimization perspective~\cite{kleinberg1999fairness,toth2002models,pioro2007fair,PollnerRSW2022,KnightMN2024}, little attention has been given to the fair treatment of gig workers.

We focus on the distribution of delivery tasks from a warehouse (the \textit{hub}) that are placed on the vertices of a graph and are connected through an edge (a route) between them. The goal is then to distribute these tasks among a fixed set of agents while adhering to given well-defined notions of fairness and economic efficiency. 

A substantial subset of these problems either excludes monetary transfers entirely (e.g., charity organizations) or involves only fixed-salary labor arrangements (e.g., postal service workers). Developing fair algorithms for such scenarios has sparked interest in designing algorithms without money~\cite{procaccia2013approximate,AshlagiS2014,NarasimhanAP2016,BalseiroGS2019,padala2021mechanism} and are notably more challenging compared to those that allow payment-based compensations (i.e., monetary transfers) based on specific tasks~\cite{nisan1999algorithmic}. Motivated by this, we primarily focus on a fairness notion of \textit{minimax share} (MMS), which aims to guarantee that no agent incurs a (submodular) cost greater than what they would receive under an (almost) equal distribution. While MMS allocations are guaranteed to exist and are compatible with the economic notion of Pareto optimality (PO), computing such allocations has been shown to be computationally intractable~\cite{HosseiniNW2024}.

\subsection{Our Contribution}

We generalize the model from the setting where the traversal of each edge costs the same to the \textit{weighted} setting, where the cost of traversing edge can differ. This significantly extends the applicability of the model, as it allows us to capture a broader variety of real-life instances. 

\paragraph{Non-Wasteful Allocations.}

We introduce a new efficiency notion called \textit{non-wastefulness}, which is partly inspired by similar notions in the literature on mechanism design for stable matching~\cite{GotoIKKYY2016,KamadaK2017,WuR2018,AzizK2019} and auctions~\cite{KawasakiBTTY2020}. Intuitively, in our context, non-wastefulness states that no delivery order can be reassigned to a different agent so that the original agent is strictly better off and the new worker is not worse off. This fundamental efficiency axiom prevents avoidable duplicate journeys---an obvious choice by delivery agents. Moreover, in contrast to Pareto optimality, it can be verified whether a given allocation is non-wasteful in polynomial time (\Cref{thm:NW:verificationAlgorithm}). Additionally, in polynomial-time, \emph{any} distribution can be turned into a non-wasteful one where no agent is worse off (\Cref{thm:NW:easyToTurnInto}). Additionally, in \Cref{sec:MMS}, we formally settle the connection between non-wastefulness and the fairness notions of MMS, and in \Cref{sec:PoS}, we study the Price of Non-wastefulness.

\paragraph{Algorithms for MMS and Non-wasteful Allocations.}
Our main technical contribution is providing a complete complexity landscape of finding MMS and non-wasteful allocations under various natural parameters. In doing so, we paint a clear dichotomy between tractable and intractable cases. 
Specifically, in \Cref{sec:InternalAndLeaves}, we show that if the number of junctions or dead-ends of the topology is bounded, then the problem can be solved efficiently in \FPT time, even for weighted instances. In \Cref{sec:AgentsOrders}, we turn our attention to the parameterization by the number of orders and the number of agents, both parameters that are expected to be small in practice. While \FPT algorithm for the former is possible even for weighted instances, for the latter, a tractable algorithm is not possible already for two agents. Also, we close an open problem of \citet{HosseiniNW2024} by showing that their \XP algorithm for the unweighted case and parameterization by the number of agents is essentially optimal. 

\paragraph{The Impact of Topology Structure.}
\Cref{sec:topologies} is then devoted to different restrictions of the topology. The most notable result here is (in)tractability dichotomy based on the~$k$-path vertex cover, where we prove the existence of \FPT algorithms for any weighted instance and~$k \leq 3$, and intractability for unweighted instances with~$k\geq 4$. 
Along the way, we identify several polynomial-time algorithms for certain graph families, such as caterpillar graphs, and additional hardness results, such as for unweighted topologies, which are in the distance one to the disjoint union of paths. 

\paragraph{Envy-Based Fairness}
We conclude the paper with a series of results regarding envy-based fairness notions such as EF and its relaxations EF1 and EFX. The main outcome here is that non-wastefulness is incompatible with these fairness notions.

\subsection{Related Work}

Fair division of indivisible items is one of the most active areas at the intersection of economics and computer science~\cite{BouveretCM2016,AmanatidisABFLMVW2023}. Different fairness notions are studied in this area, with MMS being one of the prominent ones~\cite{AmanatidisABFLMVW2023,NguyenR2023}. A relevant literature mostly focus on computational aspects~\cite{BouveretL2016,HeinenNNR2018,NguyenR2023} and existence guarantees~\cite{KurokawaPW2018}, with special focus on approximations of MMS~\cite{BarmanK2020,Xiao020H23,AkramiGST2023,ChekuriKKM2024}. Closest to our work are recent papers of \citet{li2024fair} and \citet{wang2024fair}, which also study submodular costs; however, they do not assume a graph encoding the costs.

Several works also explored fair division on graphs~\cite{ChristodoulouFKS2023,BouveretCL2019,BredereckKN2022,EibenGHO2023,BiloCFIMPVZ2022,Madathil2023,BouveretCEIP2017,TruszczynskiL2020,LiMNS2023,DeligkasEIKS2025}. The closest model to ours is the one where we have a graph over items, each agent has certain utility for every item, and the goal is not only to find a fair allocation, but each bundle must additionally form a disjoint and connected sub-graphs.

Finally, there are multiple works exploring fairness in different gig economy contexts, including food delivery~\cite{GuptaYNCRB2022,NairYGCRB2022} and ride-hailing platforms~\cite{EsmaeiliDCNSD2023,SanchezLSB2022}. Nevertheless, these papers mostly focus on experiments and neglect the theoretical study, and the models studied therein are very different from ours.

\section{Preliminaries}

We use~$\N$ to denote the set of positive integers. For an integer~$i\in\N$, we set~$[i] = \{1,2,\ldots,i\}$ and~$[i]_0 = [i] \cup \{0\}$. For a set~$S$, we let~$2^S$ be the set of all subsets of~$S$ and, for an integer~$k\in\N$, we denote by~$\binom{S}{k}$ the set of all~$k$-sized subsets of~$S$.
For detailed notations regarding computational complexity theory (classic and parameterized), we follow the monographs of \citet{AroraB2009} and \citet{CyganFKLMPPS2015}, respectively.

\paragraph{Graph Theory}
The fair delivery problem is modeled as a connected and acyclic graph, aka a \emph{tree}. 
Let~$G=(V,E)$ be a tree rooted in vertex~$r\in V$ and~$v\in V$ be its vertex. The \emph{degree} of the vertex~$v$ is~$|\{ u \mid \{v,u\} \in E \}|$ and we call~$v$ a \emph{leaf} if~$\deg(v)$ is exactly one. By~$\leaves(G)$, we denote the set of all leaves, and we set~$\numLeaves=|\leaves(G)|$. All non-leaf vertices are called \emph{inner vertices}. A vertex~$p=\operatorname{parent}(v)$ is a \emph{parent} of vertex~$v$ if~$p$ is a direct predecessor of~$v$ on the shortest~$r,v$-path, and~$\operatorname{children}(v)$ is a set of vertices whose parent is vertex~$v$. By~$G^v$, we denote the sub-tree of~$G$ rooted in vertex~$v$, and, for a set~$S\subseteq V$, we use~$W_S$ to denote the set of all shortest paths with one end in~$r$ and a second end in some vertex of~$S$.

\paragraph{Distribution of Delivery Orders.}

In \emph{distribution of delivery orders}, we are given a \emph{topology}, which is an edge-weighted tree~$G=(V,E,\wFn)$ rooted in a vertex~$h\in V$, called a \emph{hub}, and a set of agents~$\agents=\{1,\ldots,\numAgents\}$. The vertices in~$V\setminus\{h\}$ are called \emph{orders}. By~$\numOrders$, we denote the number of orders in the given instance. The goal is to find an \emph{allocation}~$\alloc\colon \orders\to\agents$. For the sake of simplicity, we denote by~$\alloc_i$ the set of orders allocated to an agent~$i$; that is,~$\alloc_i = \{ v\in\orders \mid \alloc(v) = i \}$. Moreover, we say that~$\alloc_i$ is agent~$i$'s \emph{bundle} and that an order~$v\in\alloc_i$ is \emph{serviced} by an agent~$i\in\agents$. By~$\Pi$, we denote the set of all possible allocations. Formally, an instance of our problem is a triple~$\mathcal{I}=(\agents,G,h)$. We say that an instance~$\mathcal{I}$ is \emph{unweighted} if the weights of all edges are the same. Otherwise,~$\mathcal{I}$ is weighted.

The \emph{cost} of servicing an order~$v\in \orders$, denoted~$\cost(v)$, is equal to the length of the shortest path between~$h$ and~$v$. A cost for servicing a set~$S\subseteq \orders$ is equal to the length of a shortest walk starting in~$h$, visiting all orders of~$S$, and ending in~$h$, divided by two. Observe that such a walk may also visit some orders that are not in~$S$. It is apparent that the cost function is \emph{submodular} and \emph{identical} for all agents.

\paragraph{Fairness.}

In this work, we are interested in finding \emph{fair} allocations. Arguably, the most prominent notion studied in the context of resource allocation is \emph{envy-freeness} (EF), which requires that no agent likes a bundle allocated to any other agent more than the bundle allocated to them. Formally, we define envy-freeness as follows.

\begin{definition}
	An allocation~$\alloc$ is \emph{envy-free} (EF) if for every pair of agents~$i,j\in\agents$ it holds that~$\cost(\alloc_i) \leq \cost(\alloc_j)$.
\end{definition}

Observe that since the cost functions are identical, the EF allocations are necessarily \emph{equitable}, meaning that the cost for every agent is the same. It is easy to see that such allocations are not guaranteed to exist: consider an instance with a single order and two agents.

Therefore, we will further focus on some relaxation of envy-freeness. The first relaxation we study is called \emph{envy-freeness up to one order}~(EF1) and adapts a similar concept from the fair division of indivisible items literature. Here, we allow for a slight difference between agents' costs.

\begin{definition}
	An allocation~$\alloc$ is \emph{envy-free up to one order} (EF1) if, for every pair of agents~$i,j$, either~$\alloc_i = \emptyset$ or there exists an order~$v\in \alloc_i$ such that~$\cost(\alloc_i\setminus\{v\}) \leq \cost(\alloc_j)$.
\end{definition}

We also consider \emph{minimax share guarantee} (MMS) as a desired fairness notion.\footnote{In the literature on fair division, this notion is usually studied under the name \emph{maximin share}. In our setting, however, all items have negative utility for agents, so instead of having all costs negative, we reverse the objectives and obtain an equivalent notion.} This notion can be seen as a generalization of the famous cake-cutting mechanism and requires that the cost of each agent is, at most, the cost of the worst bundle in the most positive allocation. Formally, the notion is defined as follows.

\begin{definition}
	An MMS-share of an instance~$\mathcal{I}$ of fair distribution of delivery items is defined as 
	\[
	\MMSshare(\mathcal{I}) = \min_{\alloc\in\Pi} \max_{i\in[\numAgents]} \cost(\alloc_i).
	\]
	We say that an allocation~$\alloc$ is \emph{minimax share} (MMS), if for every agent~$i\in\agents$, it holds that~$\cost(\alloc_i) \leq \MMSshare(\mathcal{I})$.
\end{definition}

Observe that since the cost functions are identical, we define the MMS-share for the whole instance and not separately for each agent.

\paragraph{Efficiency.}

We consider several notions of economic efficiency. First, we introduce a variant of utilitarian efficiency, which requires minimizing the sum of the costs of all bundles.

\begin{definition}
	An allocation,~$\alloc$, is \emph{utilitarian optimal} if for every other allocation~$\alloc'$ it holds that~$\sum_{i\in\agents} \cost(\alloc'_i) \geq \sum_{i\in\agents} \cost(\alloc_i)$.
\end{definition}

Finding a utilitarian optimal allocation is trivial; we just allocate all orders to a single agent. However, such an allocation is clearly very unfair.
Thus, we consider a weaker notion of Pareto optimality.

Informally, an allocation~$\alloc$ is Pareto optimal if there is no other allocation~$\alloc'$ such that no agent is worse in~$\alloc'$ and at least one agent is strictly better off in~$\alloc'$. It is easy to see that each utilitarian optimal is also a Pareto optimal.

\begin{definition}
	An allocation~$\alloc$ is \emph{Pareto optimal} (PO), if there is no allocation~$\alloc'$ such that for every~$i\in\agents$~$\cost(\alloc_i) \geq \cost(\alloc'_i)$ and for at least one agent the inequality is strict.
\end{definition}

The utilitarian optimal aims at maximizing the overall ``happiness'' of society. This may, however, lead to solutions where the cost for most individuals in an optimal solution is very low, and few agents carry almost the whole burden. For example, recall that allocation in which one agent services all orders is always utilitarian optimal. A different welfare perspective is egalitarian optimal solutions, where the cost for the worst-off agent(s) is minimized.

\begin{definition}
	Allocation~$\alloc$ is an \emph{egalitarian optimal} if for every other allocation~$\alloc'$ it holds that~$\max_{i\in\agents} \cost(\alloc'_i) \geq \max_{i\in\agents} \cost(\alloc_i)$. 
\end{definition}

\section{Non-wasteful Allocations}\label{sec:NW}

In this setting, some economic efficiency notions, such as utilitarian optimality, may not be generally compatible with fairness.
Moreover, computing an MMS allocation along with Pareto optimality is computationally hard~\cite{HosseiniNW2024}.
Thus, we propose a weaker efficiency notion of non-wastefulness.
Informally, a non-wasteful allocation requires that no agent~$i$ should be pushed to service an extra order if assigning this order to another agent~$j$ reduces the cost of~$i$'s bundle without increasing the cost of~$j$'s bundle.
Formally, we define our efficiency notion as follows; for an illustration of the definition, we refer the reader to  \Cref{fig:NW:illustration}.

\begin{figure}[bt!] \label{fig:nonwaste}
	\centering
	\begin{tikzpicture}
		\tikzstyle{A1} = [draw,circle,ultra thick,green!80!black,inner sep=6pt,fill=none]
		\tikzstyle{A2} = [draw,ultra thick,red!30,inner sep=8pt,fill=none]
		\tikzstyle{A3} = [draw,regular polygon,regular polygon sides=6,ultra thick,yellow!90!black,inner sep=6pt,fill=none]
		
		\node[draw,circle,fill=blue!15] (h) at (0,0) {$h$};
		\node[draw,circle] (v1) at (2,0) {};
		\node[A2] at (2,0) {};
		\node[draw,circle] (v2) at (2,1) {};
		\node[A2] at (2,1) {};
		\node[draw,circle] (v3) at (4,1) {};
		\node[A3] at (4,1) {};
		\node[draw,circle] (v6) at (6,1) {};
		\node[A1] at (6,1) {};
		\node[draw,circle] (v4) at (6,0) {};
		\node[A3] at (6,0) {};
		\node[draw,circle] (v5) at (4,0) {};
		\node[A2] at (4,0) {};
		
		\draw (h) edge (v1) edge (v2);
		\draw (v2) edge (v3);
		\draw (v1) edge (v5);
		\draw (v3) edge (v4) edge (v6);
		
		\node at (3,-1) {$\cost(\alloc_1) = \cost(\alloc_2) = \cost(\alloc_3) = 3$};
	\end{tikzpicture}
	\hspace{0.5cm}
	\begin{tikzpicture}
		\tikzstyle{A1} = [draw,circle,ultra thick,green!80!black,inner sep=6pt,fill=none]
		\tikzstyle{A2} = [draw,ultra thick,red!30,inner sep=8pt,fill=none]
		\tikzstyle{A3} = [draw,regular polygon,regular polygon sides=6,ultra thick,yellow!90!black,inner sep=6pt,fill=none]
		
		\node[draw,circle,fill=blue!15] (h) at (0,0) {$h$};
		\node[draw,circle] (v1) at (2,0) {};
		\node[A2] at (2,0) {};
		\node[draw,circle] (v2) at (2,1) {};
		\node[A1] at (2,1) {};
		\node[draw,circle] (v3) at (4,1) {};
		\node[A3] at (4,1) {};
		\node[draw,circle] (v6) at (6,1) {};
		\node[A1] at (6,1) {};
		\node[draw,circle] (v4) at (6,0) {};
		\node[A3] at (6,0) {};
		\node[draw,circle] (v5) at (4,0) {};
		\node[A2] at (4,0) {};
		
		\draw (h) edge (v1) edge (v2);
		\draw (v2) edge (v3);
		\draw (v1) edge (v5);
		\draw (v3) edge (v4) edge (v6);
		
		\node at (3,-1) {$\cost(\alloc_1) = \cost(\alloc_3) = 3 \quad \cost(\alloc_2) = 2$};
	\end{tikzpicture}
	\caption{An illustration of non-wastefulness. On the top, we depict an allocation that is not non-wasteful: while both the green (circle,~$1$) agent and yellow (diamond,~$3$) agents service a vertex if and only if they service a leaf in the respective sub-tree, the red (square,~$2$) agent services the order of the top branch even though it is not servicing any leaf of this sub-tree. On the bottom, we depict a non-wasteful allocation for the same instance. Observe that in this case, the non-wasteful allocation even strictly improved the cost for the red agent.}
	\label{fig:NW:illustration}
\end{figure}
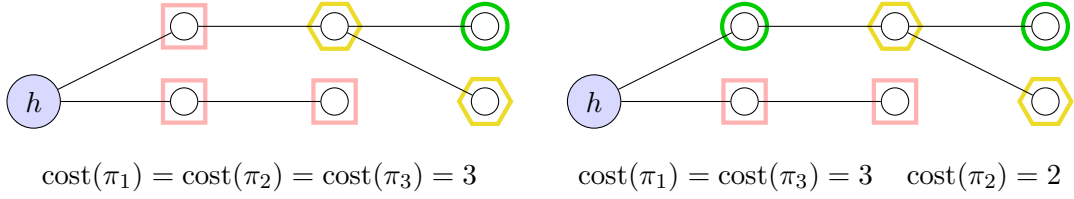

\begin{definition}\label{def:NW}
	An allocation~$\alloc$ is \emph{non-wasteful}, if for every order~$v\in\orders$ it holds that if an agent~$i\in\agents$ services~$v$, then~$i$ also services some leaf~$\ell\in\leaves(G^v)$.
\end{definition}

A Pareto optimal allocation implies non-wastefulness, but the converse does not hold. Thus, a non-wasteful allocation is always guaranteed to exist (since PO allocations always exist).

\begin{proposition}\label{thm:NW:existence}
	A non-wasteful allocation is guaranteed to exist and can be found in linear time.
\end{proposition}

Note that the same observation holds also for the utilitarian optimal (and consequently for the Pareto optimality).
An egalitarian optimal allocation always exists; however, finding such an allocation is computationally hard (see formal proof in \Cref{sec:MMS}).

\subsection{Algorithmics of Non-wastefulness}

Our first result shows that we can decide in polynomial time whether a given allocation is non-wasteful or not. This stands in direct contrast with Pareto optimality, which, under the standard theoretical assumptions, cannot admit a polynomial-time algorithm for its associated verification problem~\cite{KeijzerBKZ2009}, and makes non-wastefulness arguably one of the fundamental axioms each distribution of delivery orders should satisfy, as agents can check this property basically in hand without the need of extensive computational resources.

The results established in the remainder of this section serve as stepping stones for multiple subsequent sections, where we investigate the algorithmic aspects of non-wastefulness combined with different fairness notions. A naive procedure for verification of non-wastefulness just, for every internal vertex~$v$, checks whether at least one of the leaves in the sub-tree rooted in~$v$ is serviced by the agent servicing~$v$. %

\begin{theorem}
	\label{thm:NW:verificationAlgorithm}
	There is an algorithm that, given an instance~$\mathcal{I}$ and an allocation~$\alloc$, decides whether~$\alloc$ is a non-wasteful allocation in~$\Oh{\numOrders^2}$ time.
\end{theorem}
\begin{proof}
	The idea behind the algorithm is to store an array~$A_u$ for every order~$u\in\orders$, which keeps track of all agents servicing leaves of a sub-tree rooted in~$u$. A naive approach to compute this array for a single order, running in~$\Oh{\numOrders\cdot\numAgents}$ time, would be to do an element-wise OR of all~$A_w$, where~$w\in\operatorname{children}(u)$. Since there are~$\numOrders$ orders, the overall running time of such procedure would be~$\Oh{\numOrders^2\cdot \numAgents} = \Oh{\numOrders^3}$. In the following, we show a more efficient procedure.
	
	The initial step is to allocate array~$A_u$ for every order~$u\in\orders$ and set~$A_u[i] = \texttt{false}$ for every~$i\in\agents$. This takes~$\Oh{\numOrders\cdot\numAgents}$ time. Then, the main procedure starts. We fix an arbitrary ordering~$(\ell_1,\ldots,\ell_\numLeaves)$ of the leaves, and then we proceed with the leaves one by one according to the fixed order. For each leaf~$\ell$, we do the following. We determine the agent~$i$ servicing the leaf~$\ell$, set~$A_\ell[i] = \texttt{true}$, and propagate this information to the parent~$p=\operatorname{parent}(\ell)$. If~$p = h$ or~$A_p[ i ]$ is already set to \texttt{true}, we stop the propagation and continue with another leaf. Otherwise, we set~$A_p[ i ]$ to \texttt{true} and propagate the information further to the vertex~$p'=\operatorname{parent}(p)$. At worst, the information is propagated all the way up to the hub. Observe that using this procedure, we fill the arrays in~$\Oh{L\cdot \operatorname{depth}(G)} = \Oh{\numOrders^2}$ time, as each order of the tree is visited at most~$\numLeaves = \Oh{\numOrders}$ times. As the final step, we traverse the tree, and for each order~$u\in\orders$, we check whether~$A_u[i] = \texttt{true}$, where~$i\in\agents$ is the agent servicing~$u$ in~$\alloc$. If this holds for every order, the algorithm returns \Yes. Otherwise, if the algorithm encounters at least one order where the condition is violated, it immediately returns \No and terminates. The overall running time is~$\Oh{\numOrders\cdot\numAgents + \numOrders^2 + \numOrders} = \Oh{\numOrders^2}$.
\end{proof}

The next important property of non-wastefulness is that, given an allocation~$\alloc$, we can efficiently convert it to a non-wasteful allocation that does not differ from~$\alloc$ very significantly and while weakly improving the cost for agents. This result is appealing from the practical perspective, as it can be applied to any existing allocation of delivery tasks with negligible (polynomial) computational overhead. This clearly indicates that non-wastefulness can be very easily used as a layer on top of the current approaches (both algorithmic and manual) for the distribution of delivery tasks without affecting its computability.

\begin{theorem}
	\label{thm:NW:easyToTurnInto}
	There is a linear-time algorithm that, given an allocation~$\alloc$, returns a non-wasteful allocation~$\alloc'$ such that~$\alloc_i\cap\leaves(G) = \alloc'_i\cap\leaves(G)$ and~$\cost(\alloc'_i) \leq \cost(\alloc_i)$ for every~$i\in\agents$. In other words, in the new non-wasteful allocation~$\alloc'$, the set of leaves serviced by an agent~$i\in\agents$ remains the same as in~$\alloc$.
\end{theorem}
\begin{proof}
	The algorithm first traverses the allocation~$\alloc$ and replaces each standard agent~$i$ with its temporal twin~$\underline{i}$. Then, in several rounds, it processes the leaves one by one. Specifically, the algorithm starts with~$X = \leaves(G)$ and, while~$X$ is non-empty, it selects a leaf~$\ell$ serviced by the smallest~$\underline{i}$. For this leaf~$\ell$, the algorithm traverses the tree from~$\ell$ all the way up to the first order serviced by a standard agent or to the hub~$h$ if no such order exists. All orders, including~$\ell$, allocated to temporal agents are added to~$\alloc'_i$, the leaf~$\ell$ is removed from~$X$, and the algorithm continues with another unprocessed leaf.
	
	The algorithm is clearly finite, as the number of leaves is finite, and for each leaf, the algorithm does only finitely many steps. Also, the algorithm can be implemented in linear time: instead of having the set~$X$ and finding a leaf allocated to an agent with the smallest index~$\underline{i}$ (which would require linear look-up every time), we may have a set of sets~$X_{\underline{1}},\ldots,X_{\underline{\numAgents}}$, and iterate through them. Then, we can find the next~$\ell$ in constant time. Moreover, even though the procedure for every~$\ell$ is linear in the worst case, observe that, in fact, each vertex of~$G$ is considered exactly once during the bottom-up traversal of the tree. Overall, the algorithm can be implemented in linear time as we promised.
	
	For the correctness, assume first that~$\bigcup_{i=1}^\numAgents \alloc'_i \not= \orders$. Then, an order~$v\in\orders$ exists, which is not allocated to any agent. Let~$\ell$ be a leaf of the subtree~$T_v$ rooted in~$v$ considered the first by our algorithm. By the definition of the algorithm, before~$\ell$ is considered, all orders in~$T_v$, including~$v$, are allocated to twin agents. What the algorithm does is that it traverses the tree from~$\ell$ to the hub~$h$ until it reaches~$h$ or an order assigned to a standard agent. Since such a path necessarily visits~$v$ and all orders in~$T_v$ are initially allocated to twin agents, the algorithm added~$v$ to a bundle of some standard agents. Hence, such a~$v$ cannot exist, and all the orders are therefore allocated in~$\alloc'$.
	
	Next, we show that~$\alloc'$ is indeed a non-wasteful allocation. The condition requires that~$\alloc'_i$ contains all orders on the path from each leaf~$\ell \in \alloc'_i$ to the hub~$h$, except for those allocated to some preceding agent~$i'<i$. Observe that this is exactly what the algorithm does, as it adds all orders on a path from~$\ell\cap\alloc'_i$ to~$h$ unless it reaches~$h$ or an order is assigned to a standard agent. If the first situation appears for a leaf~$\ell$, then clearly the whole~$W_{\{\ell\}}$ is in~$\alloc'_1$. In the second case, the algorithm stopped before it reached~$h$; however, it must have reached an order~$w$ serviced by a standard agent. This agent is either~$i$ or some~$i'$. In the former case, all the orders on the path from~$w$ to~$h$ are also allocated to~$i$. Hence, the whole~$W_{\{\ell\}}$ is part of~$\alloc'_i$ even in this case. In the latter case, all orders on the path from~$w$ to~$h$ are allocated to~$i'$, which does not violate the condition of being a non-wasteful allocation. That is,~$\alloc'$ is clearly a non-wasteful allocation.
\end{proof}

\subsection{Relations to Other Efficiency Notions}

\begin{figure}
	\centering
	\begin{tikzpicture}
		\tikzstyle{A1} = [draw,circle,ultra thick,green!80!black,inner sep=6pt,fill=none]
		\tikzstyle{A2} = [draw,ultra thick,red!30,inner sep=8pt,fill=none]
		\tikzstyle{A3} = [draw,regular polygon,regular polygon sides=6,ultra thick,yellow!90!black,inner sep=6pt,fill=none]
		
		\node[draw,circle,fill=blue!15] (h) at (0,0) {$h$};
		\node[draw,circle] (v1) at (2.5,-1.5) {};
		\node[A1] at (2.5,-1.5) {};
		\node[draw,circle] (v2) at (1,-1.5) {};
		\node[A1] at (1,-1.5) {};
		\node[draw,circle] (v3) at (-1,-1.5) {};
		\node[A2] at (-1,-1.5) {};
		
		\node[draw,circle] (v4) at (3,-3) {};
		\node[A1] at (3,-3) {};
		\node[draw,circle,inner sep=1.5pt] (v5) at (2,-3) {$u$};
		\node[A3] at (2,-3) {};
		\node[draw,circle,inner sep=1.5pt] (v6) at (1,-3) {$v$};
		\node[A1] at (1,-3) {};
		\node[draw,circle] (v7) at (-2,-3) {};
		\node[A2] at (-2,-3) {};
		\node[draw,circle] (v8) at (-3,-1.5) {};
		\node[A2] at (-3,-1.5) {};
		\node[draw,circle] (v9) at (-4,-3) {};
		\node[A2] at (-4,-3) {};
		
		\draw (h) edge (v1) edge (v2) edge (v3);
		\draw (v1) edge (v4) edge (v5);
		\draw (v2) edge (v6);
		\draw (v3) -- (v7) -- (v8) -- (v9);
	\end{tikzpicture}
	\caption{An example of a non-wasteful and MMS allocation that is not Pareto optimal. We use green circles to highlight orders serviced by the agent~$1$, red squares to highlight orders serviced by the agent~$2$, and yellow diamonds for those serviced by the agent~$3$. Observe that if we move the order~$u$ from~$\alloc_3$ to~$\alloc_1$ and simultaneously the order~$v$ and~$\operatorname{parent}(v)$ from~$\alloc_1$ to~$\alloc_3$, the cost of bundle~$\alloc_3$ remains the same, while the cost of bundle~$\alloc_1$ decreases. That is, the depicted allocation is not Pareto optimal.}
	\label{fig:NW:notNecessarilyPO}
\end{figure}
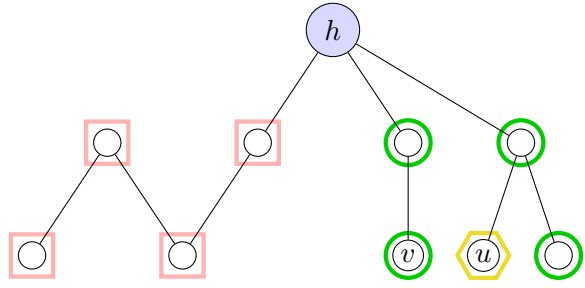

In this section, we formally settle the relation of non-wastefulness with other economic efficiency notions commonly studied in the related literature. We start with the Pareto optimality and show that PO implies non-wastefulness but not vice versa. This means that non-wastefulness is a strictly weaker notion of efficiency than Pareto optimality.

\begin{proposition}
	\label{thm:NW:ParetoOptimalityRelation}
	Every Pareto optimal allocation is also non-wasteful. On the other hand, a non-wasteful allocation exists, which is not Pareto optimal.    
\end{proposition}
\begin{proof}
	Let~$\alloc$ be an allocation that is MMS and Pareto optimal and, for the sake of contradiction, suppose that~$\alloc$ is wasteful. Then, by the definition of non-wastefulness, there exists a sub-tree~$G^v$ and an order~$u\in V(G^v)$ serviced by an agent~$i$ which is not servicing any leaf of~$G^v$. As there may be multiple such orders, let~$u$ be the one so that~$u$ is the only order serviced by~$j$ in the sub-tree~$G^u$. Clearly,~$u$ is not a leaf, so~$u$ has at least one child~$w$. Let~$j$ be an agent servicing~$w$. If we add~$u$ to the bundle~$\alloc_j$, the cost of this bundle remains the same, as the agent~$j$ anyway traverses~$u$ on its walk to~$w$. Moreover, the cost of the bundle~$\alloc_i$ is strictly decreased by the removal of~$u$, as the agent~$i$ does not have any order in~$G^u$ and hence does not visit this sub-tree now. The costs of all remaining bundles are still the same, and we have a contradiction that~$\alloc$ is Pareto optimal. Consequently, every Pareto optimal allocation is also non-wasteful. To see the second part of the statement, we refer the reader to \Cref{fig:NW:notNecessarilyPO} for an example of a non-wasteful (and MMS) allocation, which is not Pareto optimal.
\end{proof}

Since every allocation that is utilitarian optimal is also Pareto optimal, we directly obtain the following corollary.

\begin{corollary}\label{thm:NW:utilitarianOptimalityRelation}
	Every allocation~$\alloc$, which is utilitarian optimal, is also non-wasteful. On the other hand, a non-wasteful allocation exists, which is not utilitarian optimal.
\end{corollary}

To see the second part of \Cref{thm:NW:utilitarianOptimalityRelation}, recall the instance from \Cref{fig:NW:notNecessarilyPO}; the depicted allocation is non-wasteful but is not utilitarian optimal, as two distinct agents visit the parent of the vertex~$u$.

To complete the picture, our last result shows the relation between non-wasteful allocations and allocations that are egalitarian optimal.

\begin{proposition}
	\label{thm:NW:egalitarianOptimalityRelation}
	There is an allocation~$\alloc$, which is an egalitarian optimal but is not non-wasteful, and vice versa. Nevertheless, there always exists a non-wasteful egalitarian optimal allocation.
\end{proposition}
\begin{proof}
	First, recall the instance from the \Cref{fig:NW:illustration}. As the depth of the tree is three, the minimum cost an agent can have in this allocation is three. The allocation on the left is not non-wasteful, but the cost of each bundle is exactly three. Therefore, this allocation is egalitarian optimal. In the opposite direction, if we, for the same instance, allocate all orders to a single agent, the resulting allocation is non-wasteful; on the other hand, such an allocation is not egalitarian optimal as the cost of the non-empty bundle is six.
	
	For the last part of the statement, let~$\alloc$ be an egalitarian optimal allocation, which is not non-wasteful. We can use \Cref{thm:NW:easyToTurnInto} to find an allocation~$\alloc'$ such that~$\cost(\alloc'_i) \leq \cost(\alloc_i)$ for every~$i\in\agents$. Consequently,~$\alloc'$ is also an egalitarian optimal.
\end{proof}

\subsection{Price of Non-wastefulness}\label{sec:PoS}

Now, we quantify the efficiency loss we can experience if we require the allocation to be non-wasteful. We use both utilitarian and egalitarian optimals as benchmarks. Before we present our results, we introduce some more notation.

\begin{definition}
	Let~$\Pi^{\operatorname{NW}}(\mathcal{I})$ be a set of all non-wasteful allocations for an instance~$\mathcal{I}$ and~$f\colon \Pi\to\mathbb{R}$ be an efficiency measure. We define the \emph{optimistic price of non-wastefulness}~($\PoNW^+$) as 
	\[ 
	\PoNW^+_f = \sup_\mathcal{I} \frac
	{\min_{\alloc\in\Pi^{\operatorname{NW}}(\mathcal{I})}f(\pi)}
	{\min_{\alloc\in\Pi(\mathcal{I})} f(\pi)}
	\]
	and the \emph{pessimistic price of non-wastefulness} ($\PoNW^-$) as
	\[
	\PoNW^-_f = \sup_\mathcal{I}\frac
	{\max_{\alloc\in\Pi^{\operatorname{NW}}(\mathcal{I})}f(\pi)}
	{\min_{\alloc\in\Pi(\mathcal{I})} f(\pi)},
	\]
	respectively. In the remainder of this section, we are interested in either \emph{utilitarian} efficiency (UTIL), where~$f(\alloc) = \sum_{i\in\agents} \cost(\alloc_i)$, or \emph{egalitarian} efficiency (EGAL), where~$f(\alloc) = \max_{i\in\agents} \cost(\alloc_i)$.
\end{definition}

Note that the optimistic price of non-wastefulness corresponds to the price of stability~\cite{AnshelevichDKTWR2008}, and the pessimistic price of non-wastefulness corresponds to the price of anarchy~\cite{KoutsoupiasP2009}, respectively.

In \Cref{thm:NW:utilitarianOptimalityRelation} and \Cref{thm:NW:egalitarianOptimalityRelation}, we showed that for every optimal allocation, there always exists an equivalent optimal and non-wasteful allocation. Therefore, we directly obtain the following results for the optimistic price of non-wastefulness.

\begin{corollary}
	It holds that~$\PoNW^{+}_{\operatorname{UTIL}} = \PoNW^{+}_{\operatorname{EGAL}} = 1$.
\end{corollary}

The same set of results also hints at the behavior of the pessimistic variant of the price of non-wastefulness, as it will necessarily be close to the ratio between utilitarian and egalitarian optimal. The most important part of this analysis is that the pessimistic price of non-wastefulness for both benchmarks asymptotically depends only on the number of agents~$\numAgents$, which, in practical applications, can be expected to be relatively small compared to the number of orders.

\begin{theorem}
	$\PoNW^-_{\operatorname{UTIL}} = \frac{\numAgents\cdot(\numOrders-\numAgents+1)}{\numOrders}$ and~${\PoNW^-_{\operatorname{EGAL}} = \numAgents}$.
\end{theorem}
\begin{proof}
	First, we show the utilitarian case. For an upper-bound, an allocation assigning all orders to a single agent is clearly a social optimal, and the cost of such allocation is exactly~$\numOrders$. The worst non-wasteful allocation in terms of the utilitarian optimal is the one where almost all the vertices are visited by all the agents. \citet[Proposition~8]{HosseiniNW2024} showed that the worst such allocation achieves an overall cost of~$\numAgents\cdot(\numOrders-\numAgents+1)$. Now, the upper-bound follows. Next, we show that this bound is tight. Assume an instance where the topology is a path where one endpoint is the hub~$h$ and the second endpoint has additionally~$\numAgents$ leaves attached to it. The worst non-wasteful allocation~$\alloc$ allocates each leaf to a different agent. The vertices of the path can be allocated arbitrarily. Clearly, every agent must traverse all the vertices of the path and their own leaf. This leads to an overall cost of~$n\cdot(m-n+1)$. An allocation that is the best from the utilitarian perspective just allocates everything to a single agent, leading to the overall cost of~$\numOrders$.
	
	Next, we show the egalitarian case. Again, we start with the upper-bound. From the egalitarian perspective, the worst allocation allocates all orders to a single agent, achieving the maximum bundle cost of~$\numOrders$. On the other hand, an optimal egalitarian allocation may distribute the orders as evenly as possible; theoretically, optimal maximum bundle costs in such a case would be~$\numOrders/\numAgents$. To show the tightness of the construction, assume an instance where the topology is a union of~$\numAgents$ paths of the same length~$\ell$, whose one end-point is connected to the hub~$h$. If we allocate each path to a different agent, we obtain the maximum bundle cost of~$\ell = (\numAgents\cdot\ell)/\numAgents = \numOrders/\numAgents$. This shows the bound, as we can always construct an allocation assigning all orders to a single agent.
\end{proof}

\section{MMS and Non-wasteful Allocations}
\label{sec:MMS}

If we are given an MMS allocation and apply the algorithm from \Cref{thm:NW:easyToTurnInto}, we obtain a non-wasteful allocation such that the cost of no bundle is increased. Therefore, the new allocation is necessarily both MMS and non-wasteful.

\begin{proposition}\label{thm:MMS:MMSandNW:equivalent}
	Every MMS allocation can be turned into an MMS and non-wasteful allocation in linear time.
\end{proposition}

It follows from \Cref{thm:MMS:MMSandNW:equivalent} that finding MMS and non-wasteful allocations is, from the computational complexity perspective, equivalent to finding an MMS allocation. Therefore, by the result of \citet{HosseiniNW2024}, finding MMS and non-wasteful allocation is also computationally intractable, even if the instance is unweighted.

Naturally, the hardness from \citet{HosseiniNW2024} carries over to the more general weighted case, which raises the question of whether there are special topology structures or parameters for which the problem admits tractable algorithms.

In the remainder of this paper, we provide a detailed analysis of the problem's complexity, taking into account both restrictions of the topology and other natural restrictions of the input. Notably, we present the first tractable algorithms for the setting of computing fair and efficient distribution of delivery orders and, in contrast to \cite{HosseiniNW2024}, some of our positive results also apply to weighted instances, extensively broadening their practical appeal.

Before we dive deep into our results on various topologies, we show several additional auxiliary lemmas that help us simplify the proofs of the following subsections. 

\subsection{Basic Observations}

First, we show that finding MMS (and non-wasteful) allocation is as hard as deciding whether the MMS-share of an instance is at most a given integer~$q\in\N$. This follows from the fact that the cost of the most costly bundle in \emph{all} MMS allocations is the same.

\begin{lemma}
	\label{lem:family}
	Let~$\mathfrak{F}$ be a family of instances such that it is \NPh to decide whether the MMS-share of an instance from~$\mathfrak{F}$ is at most a given~$q\in N$. Then, unless~$\textsf{P}=\NP$, there is no polynomial time algorithm that finds MMS allocation for all instances from~$\mathfrak{F}$.
\end{lemma}
\begin{proof}
	Let~$\mathfrak{F}$ be a family of instances for which it is \NPh to decide whether the MMS-share is at most a given~$q\in \N$ and, for the sake of contradiction, assume that there exists an algorithm~$\mathbb{A}$ that for every instance~$\mathcal{I}\in\mathfrak{F}$ finds an MMS and non-wasteful allocation. Observe that~$\mathbb{A}$ may never return an allocation~$\alloc$ with~$\max_{i\in\agents} \cost(\alloc_i) < \MMSshare(\mathcal{I})$, as by definition, MMS-share of an instance is the maximum cost bundle in the best allocation, including the allocation~$\alloc$. Hence, in every allocation produced by~$\mathbb{A}$, there is always an agent with the bundle cost equal to the MMS-share of the instance. Now, given an instance~$\mathcal{I}=(\agents,G,h,q)$, we can use~$\mathbb{A}$ to find an MMS and non-wasteful allocation~$\alloc$ for~$(\agents,G,h)$. Then, we compare the value of the bundle of the maximum cost with the value of~$q$ and decide the instance~$\mathcal{I}$ in polynomial time. This implies that~$\mathsf{P}=\NP$, which is unlikely.
\end{proof}

The consequence of \Cref{lem:family} is that we can focus only on the complexity of deciding the MMS-share, as the impossibility of a tractable algorithm for finding MMS and non-wasteful allocations follows directly from this lemma and \Cref{thm:MMS:MMSandNW:equivalent}.

Next, we show that one can freely assume that the hub is located on some internal vertex~$v\in V(G)$. If this is not the case, then we can move the hub to the single neighbor of the leaf~$\ell=h$ and remove~$\ell$ from the instance while preserving the solution of the instance. 

\begin{lemma}
	\label{lem:MMS:ifHubIsLeafWeCanReduce}
	Let~$\mathcal{I}=(\agents,G=(V,E),h)$ be an instance of fair distribution of delivery orders such that the hub~$h$ is a leaf of~$G$ and~$\mathcal{J}$ be an instance with~$h$ removed and with the hub being~$h$'s original child~$v\in\operatorname{children}(h)$; that is,~$\mathcal{J} = (\agents,(\orders,E), v)$. Then, it holds that
	\begin{equation*}
		\MMSshare(\mathcal{I})= \MMSshare(\mathcal{J}) + \wFn(\{h,v\}).
	\end{equation*}
\end{lemma}
\begin{proof}
	Let~$\MMSshare(\mathcal{I}) = q$,~$\alloc$ be an allocation such that~$\max_{i\in[\numAgents]} \cost(\alloc_i) = q$, and~$\alloc_{\max}$ be a bundle with~$\cost(\alloc^*) = q$. Since the hub~$h$ is a leaf, for every agent~$i\in\agents$ such that~$\alloc_i \not= \emptyset$ the shortest walk~$W$ the agent~$i$ uses to service~$\alloc_i$ necessarily contains both~$h$ and~$v\in N(h)$. Moreover, each~$W$ is of the form~$h,v,\ldots,v,h$. Hence, we create an allocation~$\alloc'$ such that~$\alloc'_i = \alloc_i\setminus\{v\}$. In~$\alloc'$, the shortest walk~$W'$ for servicing~$\alloc'_i$ is the same as~$W$, just with~$h$ removed. Hence, the cost of each bundle, including~$\alloc'_{\max}$, decreases by exactly~$\wFn(\{h,v\})$. Therefore,~$\MMSshare(\mathcal{J}) \leq q - \wFn(\{h,v\})$.
	
	For the sake of contradiction, assume that~$\MMSshare(\mathcal{J}) < q -\wFn(\{h,v\})$ and let~$\alloc'$ be an allocation such that~$\max_{i\in[\numAgents]} \alloc'_i < q-\wFn(\{h,v\})$. We create an allocation~$\alloc$ such that we add~$v$ to any non-empty bundle~$\alloc'_i$. Let~$W'$ be a shortest walk servicing~$\alloc'_i$ in~$\mathcal{J}$. By definition, this walk is of the form~$v,\ldots,v$, and it holds that~$\cost(\alloc'_i) < q-\wFn(\{h,v\})$. In~$\mathcal{I}$, since the hub moved from~$v$ to~$h$ and~$v$ becomes an order, a walk~$W$ used for servicing~$\alloc_i$ now becomes~$h,v,\ldots,v,h$. That it, its length increased by twice~$\wFn(\{h,v\})$, so~$\cost(\alloc_i) = \cost(\alloc'_i) + \wFn(\{h,v\})$. That is,~$\alloc$ is an allocation showing that~$\MMSshare(\mathcal{I}) < q$, which contradicts that~$\MMSshare(\mathcal{I}) = q$.
\end{proof}

Also, by combining the negative result of {\citet[Theorem 1]{HosseiniNW2024} with \Cref{lem:MMS:ifHubIsLeafWeCanReduce}, we directly obtain that the intractability of our problem is not caused by a large number of possible routes directly leaving the hub.
	
\begin{corollary}
	Unless~$\textsf{P} = \NP$, there is no polynomial-time algorithm that finds an MMS and non-wasteful allocation, even if the instance is unweighted and the degree of the hub is~one.
\end{corollary}

\section{Small Number of Dead-ends or Junctions}
\label{sec:InternalAndLeaves}

We start our algorithmic journey with two efficient algorithms: one for topologies where the number of dead-ends (leaves) is small and one for topologies where the number of junctions (internal vertices) is small. Note that we need to study them separately as none is bounded by another. To see this, assume a star graph with one junction and an arbitrarily large number of dead-ends and, in the opposite direction, a simple path graph with exactly two dead-ends and an arbitrary number of junctions.

We start with an \FPT algorithm for the former parameter, that is, the number of leaves~$L$. The algorithm is based on the technique of dynamic programming.

\begin{theorem}
	\label{thm:MMS:FPT:numLeaves}
	When parameterized by the number of leaves~$\numLeaves$, an MMS and non-wasteful allocation can be found in \FPT time, even if the instance is weighted.
\end{theorem}
\begin{proof}
	\newcommand{\DP}{\texttt{T}}
	We prove the result by giving an algorithm running in~$2^\Oh{\numLeaves}\cdot (\numOrders+\numAgents)^\Oh{1}$ time. The algorithm is based on a dynamic programming approach, and, maybe surprisingly, it does not exploit the topology's structure, as is common for such algorithms, but rather tries to guess for each agent the set of leaves he or she is servicing in an optimal solution. The crucial observation here is that for MMS and non-wastefulness, the agents are interested only in their own bundles. Therefore, we do not need to store the whole partial allocation; rather, we need only the bundle of the currently processed agent and the list of all already allocated orders.
	
	\begin{algorithm}[tb]
		\caption{A dynamic programming algorithm for the computation of an MMS and non-wasteful allocation on instances with a small number of dead-ends.}
		\label{alg:algorithm}
		\textbf{Input}: A problem instance~$\mathcal{I}=(G,h,\agents)$.\\
		\textbf{Output}:~$\MMSshare(\mathcal{I})$.
		\begin{algorithmic}[1] %
			\State \Return~$\min\limits_{Q\subseteq\leaves(G)} \Call{SolveRec}{\numAgents,\leaves(G)\setminus Q,Q}$ \label{alg:algorithm:recInit}
			
			\Function{SolveRec}{$i,P,Q$}
			\If{$i=1$ \textbf{and}~$\DP[i,P,Q] = \texttt{undef}$}
			\If{$P = \emptyset$}
			\State~$\DP[i,P,Q] \gets \cost(Q)$
			\Else
			\State~$\DP[i,P,Q] \gets \infty$
			\EndIf
			\ElsIf{$\DP[i,P,Q] = \texttt{undef}$}
			\If{$P\cap Q = \emptyset$}
			\State~$x \gets \min\limits_{P'\subseteq P} \Call{SolveRec}{i-1,P\setminus P',P'}$ \label{alg:algorithm:recStep}
			\State~$\DP[i,P,Q] \gets \max\{ x, \cost(Q) \}$
			\Else
			\State~$\DP[i,P,Q] \gets \infty$
			\EndIf
			\EndIf
			
			\State \Return~$\DP[i,P,Q]$
			\EndFunction
		\end{algorithmic}
	\end{algorithm}
	
	More formally, the core of the algorithm is a dynamic programming table~$\DP[i,P,Q]$, where
	\begin{itemize}
		\item~$i\in\agents$ is the last processed agent,
		\item~$P\subseteq \operatorname{leaves}(G)$ is a subset of leaves allocated to agents~$1,\ldots,i-1$, and
		\item~$Q \subseteq \operatorname{leaves}(G)\setminus P$ is a bundle of agent~$i$,
	\end{itemize}
	and in each cell of~$\DP[i,P,Q]$, we store the minimum of the maximum-cost bundle over all partial allocations, where leaves of~$Q$ are assigned to agent~$i$, leaves of~$P$ are distributed between agents~$1,\ldots,i-1$, and leaves of~$\orders\setminus(P\cup Q)$ are unassigned.  The computation is then defined as of \Cref{alg:algorithm}. Note that, for the sake of exposition, the code presented computes just the optimal cost. To extend the algorithm so that it also finds an MMS and non-wasteful allocation, we store in each cell a pair~$(q,\alloc)$, where~$q$ is the minimum cost and~$\alloc$ is a partial allocation achieving this cost.
	
	The number of cells of the dynamic programming table is~$\Oh{\numAgents\cdot 2^L\cdot 2^L}\in 2^\Oh{L}\cdot n^\Oh{1}$, and each cell is computed exactly once. The most time-consuming operations of the algorithm are lines~\ref{alg:algorithm:recInit} and~\ref{alg:algorithm:recStep}, where we, at worst, try all possible subsets of leaves. That is, the overall running time of the algorithm is~$2^\Oh{\numLeaves}\cdot (\numAgents+\numOrders)^\Oh{1}$ as promised. Note that we made no assumptions about the edge weights.
	
	We call the triple~$(i,P,Q)$ a \emph{signature}. We say that a \emph{partial allocation}~$\alloc = (\alloc_1,\ldots,\pi_i)$ realizes signature~$(i,P,Q)$ if 
	\begin{enumerate}
		\item the orders in~$Q$ are allocated to the agent~$i$, that is,~$Q\subseteq \alloc_i$,
		\item the orders in~$P$ are all serviced by the agents~$1,\ldots,i-1$, that is,~$P \subseteq \bigcup_{j=1}^{i-1} \alloc_j$,
		\item an order~$v$ is serviced by an agent~$j\in[i]$ if and only~$v$ is on a path between some leaf~$\ell\in P\cup Q$ and~$h$; formally,~$\bigcup_{j=1}^{i} \alloc_j = W_{P\cup Q}$, and
		\item~$\alloc$ is a non-wasteful allocation.
	\end{enumerate}
	By~$\Pi_{i,P,Q}$, we denote the set of all allocations realizing the signature~$(i,P,Q)$. For a signature~$(i,P,Q)$ the dynamic programming table stores the value~$\min_{\alloc\in\Pi_{i,P,Q}} \max_{j\in[i]} \cost(\alloc_j)$ or~$\infty$, if~$\Pi_{i,P,Q} = \emptyset$. Note that it may happen that~$\Pi_{i,P,Q}$ is an empty set, albeit an MMS allocation is always guaranteed to exist. Such a situation occurs when one or more conditions on the partial solution cannot be satisfied. The most immediate example is when~$P\cap Q \not= \emptyset$.
	
	If~$i=1$, the set~$P$ must always be empty, as there are no previously processed agents. Therefore, the value stored in the table depends solely on the set of visited leaves. Formally, we set
	\[
	\DP[i,P,Q] = \begin{cases}
		\cost(Q)  & \text{if } P = \emptyset \text{ and}\\
		\infty    & \text{otherwise.}
	\end{cases}
	\]
	In this case, it is easy to see that there are~$2^L\cdot 2^L = 4^L$ different signatures, and for each signature, the value can be computed in polynomial time.
	
	Now, let~$i> 1$. Then, the stored value depends on the bundle of the currently processed agent~$i$ and also on the bundle of all previous agents. However, we do not need to know the exact bundles of all previously processed agents; we are interested only in the leaves in the bundle of the current agent and in the best possible (in terms of the MMS-share) partial partition of orders in~$P$. Formally, the computation is defined as follows.
	\[
	\DP[i,P,Q] = \begin{cases}
		\min\limits_{P'\subseteq P}\max\{ \DP[i-1,P\setminus P', P'], \cost(Q) \} & \\
		\phantom{XXXXx}\text{if } P\cap Q = \emptyset \text{ and}\\
		\infty %
		\phantom{mmmn}\text{otherwise.}
	\end{cases}
	\]
	For every~$i > 1$, there are~$2^\numLeaves\cdot2^\numLeaves = 4^\numLeaves$ possible signatures, and for each signature, the value of the corresponding cell can be computed in~$2^\numLeaves\cdot (\numOrders+\numAgents)^\Oh{1}$ time. Overall, the computation for every~$i > 1$ takes~$2^\Oh{\numLeaves}\cdot (\numOrders+\numAgents)^\Oh{1}$ time.

	Now, we show that the computation is indeed correct. First, we show that once the dynamic programming table is computed, a corresponding partial allocation always exists for each cell, storing a value different from infinity.
	
	\begin{claim}\label{clm:MMS:FPT:numLeaves:inductionStep}
		For every~${\DP[i,P,Q] \not= \infty}$, there always exists a partial allocation~${\alloc\in\Pi_{i,P,Q}}$ such that \[\max_{i\in\agents} \cost(\alloc_i) = \DP[i,P,Q].\]
	\end{claim}
	\begin{claimproof}
		We show the claim by the mathematical induction over~$i$. First, let~$i=1$. In this case, the set~$P$ is always empty, and we can focus only on~$Q$. Let~$Q$ be a fixed set of leaves. We construct a partial allocation~$\alloc = (\alloc_1)$ such that we set~$\alloc_1 = W_Q\setminus\{h\}$. Then,~$\alloc$ clearly realizes~$(i,P,Q)$, as~$Q=\alloc_1$,~$P = \emptyset$, and~$\bigcup_{j=1}^i \alloc_j = W_{P\cup Q} = W_Q = \alloc_1$. Moreover,~$\alloc$ is non-wasteful since it contains only the orders on the paths between leaves of~$Q$ and the hub~$h$. By the same argument, the cost of~$\alloc_1$ is equal to the cost of servicing~$Q$. This is exactly the value stored in~$\DP[i,P,Q]$, and therefore, the claim holds for the basic case.
		
		Now, let~$i > 1$ and assume that the claim holds for~$i-1$. Let~$\DP[i,P,Q] = q$ for some~$q\not=\infty$. Since~$q\not=\infty$, it holds that~$P\cap Q = \emptyset$. Let~$P'\subseteq P$ be a set such that~$\DP[i-1,P\setminus P',P'] = w'$ is minimized. By the induction hypothesis, there exists a partial allocation~$\alloc'$ realizing the signature~$(i-1,P\setminus P', P')$ such that~$\max_{i\in\agents} \cost(\alloc_i) = w'$. The partial partition~$\alloc'$ allocates, in a non-wasteful way, all orders of~$P$. We create a partial allocation~$\alloc$ such that~$\alloc = (\alloc'_1,\alloc'_2,\ldots,\alloc'_{i-1},W_Q\setminus W_P)$. It is obvious that~$\alloc$ satisfies properties (1) and (2). Moreover, let~$v$ be an order. If~$v\in W_{P}$, then~$v$ was already services by some agent in~$\alloc'$. If~$v\in W_Q\setminus W_P$, then~$v$ is serviced by the agent~$i$. No other order is serviced in the partial allocation, so~$\alloc$ satisfies the property (3). For the last property, observe that all bundles~$\alloc_j$,~$j\in[i-1]$, are non-wasteful, and~$\alloc_i$ is also non-wasteful as we added only orders of~$W_Q$. That is,~$\alloc$ is a partial allocation realizing signature~$(i,P,Q)$. By the definition of the computation,~$w$ is~$\max_{j\in[i-1]} \{\alloc'_j,\cost(Q)\}$, which is exactly the maximum cost of a bundle of~$\alloc$. This finishes the proof.
	\end{claimproof}
	
	To finalize the correctness of the computation, we show that the stored value is indeed minimal possible.
	
	\begin{claim}\label{clm:MMS:FPT:numLeaves:inductionStep}
		For every~$\alloc\in \Pi_{i,P,Q}$, it holds that the value stored in~$\DP[i,P,Q]$ is at most~$\max_{i\in\agents} \cost(\alloc_i)$.
	\end{claim}
	\begin{claimproof}
		Again, the proof is obtained by induction over values of~$i$. Let~$i=1$ and~$(1,P,Q)$ be a signature. Since~$\DP[1,P,Q] \not=\infty$, it must hold that~$P = \emptyset$ by the definition of the computation. Then, for an arbitrary subset of leaves~$Q\subseteq\leaves(G)$, there exists only one partial allocation~$\alloc$. This allocation necessarily allocates all leaves to agent~$1$, and the cost of this allocation is~$\cost(Q)$. By the definition of the computation, this is exactly the value we store in the dynamic programming table for~$(1,P,Q)$. Hence, the claim is, in this case, correct.
		
		Now, let~$i \geq 2$ and let the claim hold for~$i-1$. Since~$\alloc$ realizes~$(i,P,Q)$, it means that~$\alloc$ is non-wasteful,~$Q\subseteq\alloc_i$,~$P\subseteq\bigcup_{j=1}^{i-1} \alloc_j$, and~$\bigcup_{j=1}^{i} \alloc_j = W_{P\cup Q}$. Let~$\alloc'=(\alloc_1,\ldots,\alloc_{i-1})$ be the allocation~$\alloc$ with the bundle of agent~$i$ removed. By the induction hypothesis, the value stores in~$\DP[i-1,P\setminus P',P']$, where~$P' = \alloc_{i-1}\cap \leaves(G)$, is at most~$\max_{j\in[i-1]} \cost(\alloc_i)$. By the definition of the computation, the value stored in~$\DP[i,P,Q]$ is at most~$\max\{\DP[i-1,P\setminus P',P'],\cost(Q)\}$, showing the claim.
	\end{claimproof}
	
	Once the dynamic programming table is correctly computed, we just find~$Q\subseteq \operatorname{leaves}(G)$ such that~$\DP[n,\operatorname{leaves}(G)\setminus Q, Q]$ is minimized. By the definition of MMS-share, any corresponding partial allocation is MMS.
	
	The overall running time of the algorithm is~$2^\Oh{\numLeaves}\cdot (\numAgents+\numOrders)^\Oh{1}$, which is clearly in \FPT. Also, note that we made no assumptions about the edge weights. %
\end{proof}

The structural counterpart of the number of leaves is the number of internal vertices. Again, we show that under this parameterization, our problem is in the complexity class \FPT. However, this algorithm is completely different from the previous one and combines an insight into the structure of MMS and non-wasteful solutions with careful guessing and ILP formulation of the carefully designed subproblem. 

\begin{theorem}
	\label{thm:MMS:FPT:numInternal}
	When the instance is parameterized by the number of internal vertices~$k$ and the number of different edge weights~$\psi$, an MMS and non-wasteful allocation can be found in \FPT time.
\end{theorem}
\begin{proof}
	Our algorithm combines several ingredients. First, we show a structural lemma that allows us to restrict the number of important agents in terms of the number of internal vertices. Then, for these important agents, we guess their bundles in an optimal solution. Finally, for each guess, we design an integer linear program (ILP) that helps us verify whether our guess is indeed a solution. For the sake of exposition, we first show the proof for the unweighted instances; the generalization to instances with a bounded number of different weights is described at the end of the proof.
	
	Let~$\equiv$ be an equivalence relation over the set of leaves such that for a pair~$\ell,\ell \in \leaves(G)$ it holds that~$\ell\equiv\ell'$ if and only of~$\parent(\ell) = \parent(\ell')$. Observe that the relation partitions the leaves into~$k$ equivalence classes; we denote them~$T_1,\ldots,T_k$. In the following lemma, we show that for each allocation~$\alloc$, there exists an allocation~$\alloc'$ where no agent is worse off and which possesses a nice structure.
	
	\begin{lemma}
		\label{lem:niceAlloc}
		Let~$\alloc$ be an allocation. There always exists a \emph{nice} allocation~$\alloc'$ such that~$\cost(\alloc'_i) \leq \cost(\alloc_i)$ for every~$i\in\agents$. An allocation is~$\alloc'$ is nice if for each pair of distinct agents~$i,j\in\agents$ there exists at most one type~$t\in[k]$ so that~$|\alloc'_i\cap T_t| > 0$ and~$|\alloc'_j\cap T_t| > 0$.
	\end{lemma}
	\begin{claimproof}
		If~$\alloc$ satisfies the condition, then we are done. Otherwise, let there exist a pair of distinct agents~$i,j\in\agents$ such that there are~$t,t'\in[k]$ so that~$|\alloc_i\cap T_t|,|\alloc_i \cap T_{t'}|,|\alloc_j\cap T_t|,|\alloc_j\cap T_{t'}| > 0$. Let~$\ell'\in T_{t'}$ be an order serviced by the agent~$i$ and~$\ell\in T_t$ be an order serviced by the agent~$j$. We create~$\alloc'$ such that we move~$\ell'$ from~$\alloc_i$ to~$\alloc_j$,~$\ell$ from~$\alloc_j$ to~$\alloc_i$, and we keep all the remaining bundles the same. By our assumptions,~$\alloc_i$ contains at least one other leaf of~$T_t$ and~$\alloc_j$ at least one other leaf of~$T_{t'}$. Consequently, the order added to~$\alloc_i$ increases the cost for~$i$ by exactly one, and the removal of~$\ell'$ from~$\alloc_i$ decreases the cost for~$i$ by at least one. Hence, we obtain that~$\cost(\alloc'_i) \leq \cost(\alloc_i)$ and symmetrically for~$j$. The cost for all the other agents remains the same, as their bundles have not been changed. We continue such swaps as long as~$i$ services at least one order of~$T_{t'}$ or~$j$ services at least one order of~$T_t$. That is, after~$\Oh{m}$ steps, we obtain an equivalent nice allocation~$\alloc'$.
	\end{claimproof}
	
	The previous lemma implies that there is always an allocation, namely the nice one, where most agents service leaves of exactly one type. To see this, assume that a nice allocation~$\alloc$ exists with~$\binom{k}{2} + 1$ agents servicing at least two different types of leaves. Then, by the Pigeonhole principle, there is necessarily a pair of agents~$i$ and~$j$ both servicing at least one leaf of some~$T_t$ and~$T_{t'}$ with~$t\not=t'$, which contradicts that~$\alloc$ is nice. Consequently, at most~$\binom{k}{2}$ agents service leaves of multiple different types, and all other agents services leaves of exactly one type.
	
	In the next phase of the algorithm, we first guess the number~$\eta \leq \min\{\binom{k}{2},\numAgents\}$ of important agents, and then for each of agents~$i\in[\eta]$, we guess the structure of their bundle. Specifically, for each agent~$i\in[\eta]$, the bundle structure is a subset~$L_i\subseteq [k]$, where~$t\in L_i$ represents that, in a solution~$\alloc$, the agent~$i$ services at least one leaf of type~$t$. By \Cref{lem:niceAlloc}, we can assume that all remaining agents~$j\in[\eta+1,\numAgents]$ are servicing exactly one type of leaves, so we do not need to guess their structure.
	
	To verify whether our guess is correct, we use an integer linear programming formulation of the problem. Before introducing the problem's ILP encoding, we guess the MMS-share~$q$ of the instance. Note that since the instance is unweighted, there is only a linear number of possible values of~$q$, and we can try all of them in increasing order to obtain the minimum possible~$q$.
	
	In the formulation, we have a non-negative integer variable~$x_i^t$ for every~$i\in[\eta]$ and every~$t\in |T_i|$ representing the number of additional leaves of type~$t$ the agent~$i$ services. Additionally, we have~$k$ variables~$y_1,\ldots,y_k$ where each~$y_j$ represents the number of agents servicing only the leaves of type~$T_j$. The constraints of the program are as follows (we use~$d_t = \dist(\parent(T_t),h)$).
	\begin{align}
		\forall i\in[\eta]\colon&& \sum_{t\in L_i} ( x_i^t + 1 + d_t ) &\leq q &&\label{eq:MMS:FPT:3pvc:agentConstraint}\\
		\forall t\in[k]\colon&& \sum_{i\in[\eta]\colon t\in L_i} ( x_i^t + 1) + y_t\cdot(q - d_t) &= |T_t|\label{eq:MMS:FPT:3pvc:typeConstraint}\\
		&& \sum_{t\in[k]} y_t + \eta &\leq \numAgents \label{eq:MMS:FPT:internal:numAgentsConstraint}
	\end{align}
	
	The constraints \eqref{eq:MMS:FPT:3pvc:agentConstraint} ensure that the cost of no bundle exceeds the guessed value of the MMS-share. The constraints~\eqref{eq:MMS:FPT:3pvc:typeConstraint} then secure that all orders are serviced. Finally, due to the constraint \eqref{eq:MMS:FPT:internal:numAgentsConstraint}, the number of agents is correct. Also, observe that we do not use any objective function, as we are only interested in the feasibility of our program. However, we could exploit the objective function to, e.g., find MMS and non-wasteful allocation that minimizes the sum of costs.
	
	For the correctness, assume first that there is an MMS and non-wasteful allocation~$\alloc$. If~$\alloc$ is not a nice allocation, we use \Cref{lem:niceAlloc} to find an equivalent nice allocation in polynomial time. Therefore, assume that~$\alloc$ is a nice allocation. Also, we rename the bundles so that agents~$1,\ldots,\eta$ are servicing leaves of multiple types and agents~$[\eta+1,\numAgents]$ are servicing leaves of only one type (or no type at all). Now, we set the variables of the ILP as follows. We set each~$x_i^t$ to the value~$\{|T_t\cap \alloc_i|-1\}$ for every~$i\in[\eta]$ and every~$i\in L_t$. For the remaining variables, we assign~$y_t = \{ i\in[\eta+1,\numAgents] \mid |T_t| \cap \alloc_i \not= \emptyset \}$. Clearly, the sum of all~$y_t$ is at most~$n-\eta$, as each agent~$i\in[\eta+1,\numAgents]$ services exactly one leaf type. For the sake of contradiction, assume that \eqref{eq:MMS:FPT:3pvc:agentConstraint} is not satisfied. Then there exists an agent~$i\in[\eta]$ such that~$\sum_{t\in L_i} x_i^t + 1 + d_t = \sum_{t\in L_i} |T_t\cap \alloc_i| + d_t = \cost(\alloc_i) > q$, which contradicts that~$\alloc$ is an MMS allocation. Finally, assume that \eqref{eq:MMS:FPT:3pvc:typeConstraint} is violated with our variables assignment. Then there exists a type~$t\in[k]$ such that the number of leaves of this type services by an important agent together with the number of leaves of this type covered by other agents is not equal to~$|T_t|$. However, this implies that there exists a leaf~$\ell\in T_t$ not serviced by any agent in~$\alloc$, contradicting that~$\alloc$ is a complete allocation. Thus, this never happens, and the ILP is indeed feasible.
	
	In the opposite direction, let~$\vec{s}=(x^1_1,\ldots,x^\eta_k,y_1,\ldots,y_k)$ be a feasible solution for the ILP program we constructed. We create an allocation~$\alloc$ as follows. For every~$i\in[\eta]$, we add to~$\alloc_i$ exactly~$x_i^t + 1$ unallocated leaves of type~$t\in[k]$ together with all so-far unallocated internal vertices between~$h$ and~$\parent(T_t)$. Next, for each~$t\in[k]$, we add to the allocation~$y_t$ bundles, each consisting of~$q-d_t$ unallocated leaves of type~$t$ together with all unallocated internal vertices on the path between~$\parent(T_t)$ and~$h$. If there is an agent with an undefined bundle, then the bundle of this agent is an empty set. Clearly, the cost of every~$\alloc_j$,~$j > \eta$, is at most~$q$ and~$\alloc$ is non-wasteful. First, we verify that the number of bundles in~$\alloc$ is exactly~$n$. Obviously, the number of bundles is never smaller than~$\numAgents$, as if yes, then we complete the allocation with an appropriate number of empty bundles. Therefore, assume that the number of bundles exceeds~$\numAgents$. Based on the solution, we created~$\eta + \sum_{t\in[k]} y_t$ bundles; if this number is greater than~$\numAgents$, then the constraint~$\eqref{eq:MMS:FPT:internal:numAgentsConstraint}$ is violated, contradicting that~$\vec{s}$ is a solution. Now, assume that there exists an agent~$i\in[\eta]$ such that~$\cost(\alloc_i) > q$. We defined~$\alloc_i$ such that it contains~$x_i^t + 1$ leaves of each type~$t\in L_i$ together with some internal vertices. Since the allocation is non-wasteful, the cost of each bundle cannot be decreased by the removal of an internal vertex. Hence, we can decompose~$\cost(\alloc_i)$ as~$\sum_{t\in L_i} x_i^t + 1 + d_t$. However, since~$\vec{s}$ is a feasible solution, this value is never greater than~$q$, unless \eqref{eq:MMS:FPT:3pvc:agentConstraint} is violated. That is, the cost of each bundle is at most~$q$. Finally, let there be a leaf~$\ell\in T_t$,~$t\in[k]$, which is not part of any bundle, that is,~$\sum_{i\in[\numAgents]} |\alloc_i \cap T_t| < |T_t|$. We can decompose the left side of the inequality as~$\sum_{i\in[\eta]\colon t\in L_i} (x_i^t) + y_t\cdot(q-d_t) < |T_t|$, and we directly obtain violation of \eqref{eq:MMS:FPT:3pvc:typeConstraint}. Hence, all leaves are serviced, and therefore, by the definition of~$\alloc$, also all internal vertices are serviced. That is,~$\alloc$ is an MMS and non-wasteful allocation, finishing the correctness of the algorithm.
	
	For the running time, observe that the number of variables of the program is~$\eta\cdot k+k \in \Oh{k^2\cdot k + k} \in \Oh{k^3}$. Therefore, the program can be solved in time~$k^\Oh{k^3}\cdot \numOrders^\Oh{1}$ by the result of \mbox{\citet{Lenstra1983}}. There are~$2^\Oh{k^3}$ different guesses we need to verify, and therefore, the overall running time of the algorithm is~$2^\Oh{k^3} \cdot 2^\Oh{k^3\log k} \cdot \numOrders^\Oh{1} \in 2^\Oh{k^3\log k}\cdot\numOrders^\Oh{1}$, which is indeed in \FPT.
	
	Assume now that there is a bounded number~$\psi$ of different edge-weights. Then, we slightly change the definition of the relation~$\equiv$ so that two leaves~$\ell$ and~$\ell'$ are equivalent if and only if~$\parent(\ell) = \parent(\ell')$ and~$\wFn(\ell,\parent(\ell)) = \wFn(\ell',\parent(\ell'))$. Observe that the number of different equivalence classes is now~$\Oh{\psi\cdot k}$ instead of~$\Oh{k}$ and we use~$T_1^j,\ldots,T_k^j$,~$j\in[\psi]$, to refer to the equivalence classes of~$\equiv$. Then, we generalize \Cref{lem:niceAlloc} as follows.
	
	\begin{lemma}
		Let~$\alloc$ be an allocation. There always exists a \emph{nice} allocation~$\alloc'$ such that~$\cost(\alloc'_i) \leq \cost(\alloc_i)$ for every~$i\in\agents$. An allocation~$\alloc'$ is nice if for each pair of distinct agents~$i,j\in\agents$ and for every~$w\in[\psi]$ there exists at most one~$t\in[k]$ so that~$|\alloc'_i\cap T_t^w| > 0$ and~$|\alloc'_j\cap T_t^w| > 0$.
	\end{lemma}
	
	Observe that since we are always replacing leaves of the same weight, the proof remains completely the same. Also, the lemma implies a bound on the number of agents in terms of our parameters: namely, if for some weight value~$w\in[\psi]$ there are more than~$\binom{k}{2}$ agents servicing at least two different types of leaves, then at least two of them are servicing the same pair of~$T_t^w$ and~$T_{t'}^w$ and hence, we can apply the previous lemma to get rid of that. Consequently, there are~$\psi\cdot \binom{k}{2}$ important agents, and all the remaining agents are servicing at most one leaf type for each weight. The number of standard agents is, therefore, at most~$(k+1)^\psi$, and we have a variable~$y_j$ for every one of them. For every~$j\in[(k+1)^\psi]$, we use~$Y_j$ to represent the set of leaf types, the standard agent of type~$j$ services. We need to incorporate the weights into the ILP to complete the algorithm. This can be done by changing the conditions as follows (we set~$w_i = \wFn(\{L_i,\parent(L_i)\})$):
	\begin{align*}
		\forall i\in[\eta]\colon&& \sum_{t\in L_i} ( (x_i^t + 1)\cdot w_i + d_t ) &\leq q\\
		\forall t\in[\psi k]\colon&& \sum_{i\colon t\in L_i} ( x_i^t + 1 )\cdot w_i + \sum_{j\colon t\in Y_j}y_t\cdot\lfloor\frac{q - d_t}{w_i}\rfloor &= |T_t|\\
		&& \sum_{j} y_j + \eta &\leq \numAgents
	\end{align*}
	This time, we cannot try all possible values of~$q$, as the upper-bound on the MMS-share can be exponential in the number of orders at worst. However, we can find the correct value of~$q$ by a binary search in the interval~$[\max_{v\in\orders}\dist(v,h),\numOrders\cdot\max_{e\in E}\wFn(w)]$. The proof of correctness is then analogous to the unweighted setting, and the number of variables in the ILP formulation is again bounded in terms of our parameters. Hence, the result follows.
\end{proof}

To finalize the complexity picture with respect to the number of internal vertices, in our next result, we show that the parameter the number of different weights cannot be dropped while keeping the problem tractable; in particular, we show that if the number of edge-weights is not bounded, then an efficient algorithm cannot exist already for topologies with a single internal vertex. The reduction is from the \probName{$3$-Partition} problem~\cite{GareyJ1975}.

\begin{theorem}\label{thm:MMS:NPh:stars}
	Unless~$\textsf{P}=\NP$, there is no polynomial-time algorithm that finds an MMS and non-wasteful allocation, even if~$G$ is a weighted star and the weights are encoded in unary\footnote{We say that the weights are unary encoded if~$\wFn(w) \in \Oh{\numOrders + \numAgents}$ for every~$e\in E$.}.
\end{theorem}
\begin{proof}
	We show the hardness by a polynomial reduction from the \probName{$3$-Partition} problem. In this problem, we are given a multi-set~$A=\{a_1,\ldots,a_{3k}\}$ of integers such that~$\sum_{i\in[k]} a_i = k\cdot B$ for some~$B\in\N$, and our goal is to decide whether there exists a~$k$-sized set~$\mathcal{S}\subseteq \binom{A}{3}$ such that for every~$S\in\mathcal{S}$ we have~$\sum_{a_i\in S} a_i = B$ and~$\bigcup_{S\in\mathcal{S}} S = A$. It follows that each element~$a_i\in A$ is a member of exactly one set~$S\in\mathcal{S}$. The problem is known to be \NPc even if~$B/4 < a_i < B/2$ for every~$a_i\in A$~\cite{GareyJ1975}.
	
	Let~$\mathcal{J}$ be an instance of the \probName{$3$-Partition} problem. We create an equivalent instance~$\mathcal{I}$ of our problem as follows. First, we construct the graph~$G$. For every element~$a_i\in A$, we create one order~$v_i$ and connect it by an edge of weights~$a_i$ with the hub~$h$. Clearly, the graph is a star with~$h$ being its center. Next, we set~$N = [k]$, and we ask whether the MMS-share is at most~$q = B$.
	
	For correctness, let~$\mathcal{J}$ be a \Yes-instance and let~$\mathcal{S}=(S_1,\ldots,S_k)$ be a solution. We create an allocation~$\alloc$ so that for every~$S_i=(a_{i_1},a_{i_2},a_{i_3})$, we set~$\alloc_i = \{v_{i_1},v_{i_2},v_{i_3}\}$. Now, we claim that the worst bundle is of cost~$q$, i.e.,~$\max_{i\in[k]} \cost(\alloc_i) \leq q = k\cdot B$. Let~$\alloc_i$ be a bundle of the maximum cost across all bundles of the allocation~$\alloc$ and, for the sake of contradiction, assume that~$\cost(\alloc_i) > q$. By the construction of~$\alloc$, it holds that~$|\alloc_i| = 3$. So, suppose that~$\alloc_i = \{v_{i_1},v_{i_2},v_{i_3}\}$. By construction,~$\cost(\alloc_i) = \wFn(\{h,v_{i_1}\}) + \wFn(\{h,v_{i_2}\}) + \wFn(\{h,v_{i_3}\}) = a_{i_1} + a_{i_2} + a_{i_3}$. However,~$\cost(\alloc_i) > q$ implies that~$\sum_{a\in S_i} a > q$, which contradicts that~$\mathcal{S}$ is a solution for~$\mathcal{J}$. Therefore,~$\cost(\alloc_i) \leq q$, showing that~$\MMSshare(\mathcal{I}) \leq q$.
	
	In the opposite direction, let~$\mathcal{I}$ be an instance such that the MMS-share of this instance is at most~$q$, and~$\alloc$ be an allocation with~$\max_{i\in[k]} \cost(\alloc_i) \leq q$. 
	First, suppose that there exists a bundle~$\alloc_i$ such that~$\cost(\alloc_i) < q$. Then, it holds that~$\sum_{i\in[k]}\cost(\alloc_i) \leq (k\cdot q) - 1$; however~$\cost(\orders )= k\cdot q$. That is, at least one order is not serviced by~$\alloc$, and we obtain a contradiction with~$\alloc$ being an allocation. Consequently, the cost of each bundle is exactly~$q$. 
	Now, suppose that there is an agent~$i$ servicing only two orders. Since~$\wFn{\{h,v_j\}} < B/2 = q/2$ for every~$j\in[3k]$, we obtain that~$\cost(\alloc_i) < B/2 + B/2 < q$, which contradicts that the cost of each bundle is exactly~$q$. Hence, each agent services at least three orders. By simple counting argument, we obtain that each agent services exactly~$3$ orders, as there are~$k$ candidates and~$3k$ orders. Now, we create a set~$\mathcal{S} = (S_1,\ldots,S_k)$ such that~$S_i = \{ a_{j} \mid v_{j} \in \alloc_i \}$, and claim that~$\mathcal{S}$ is a solution for~$\mathcal{J}$. Clearly, each~$S_i$ of size three and~$\bigcup_{i=1}^{k} S_i = A$. What remains to show is that~$\sum_{a\in S_i} a = B$ for every~$S_i\in \mathcal{S}$. For the sake of contradiction, let there be a set~$S_i$ such that~$\sum_{a\in S_i} a \not= B$. Then, the cost of corresponding bundle~$\alloc_i$ is~$\cost(\alloc_i) = \wFn(\{h,v_{i_1}\}) + \wFn(\{h,v_{i_2}\}) + \wFn(\{h,v_{i_3}\}) = a_{i_1} + a_{i_2} + a_{i_3} \not= B$; however~$B=q$ and we already showed that the cost of every bundle in every solution is exactly~$q$. That is, for every~$S_i\in\mathcal{S}$ it holds that~$\sum_{a\in S_i} a = B$, which shows that~$\mathcal{S}$ is indeed a solution for~$\mathcal{J}$.
\end{proof}

\section{Small Number of Agents or Orders}
\label{sec:AgentsOrders}

In real-life instances, especially those related to applications such as charity work, it is reasonable to assume that the number of orders or the number of agents is relatively small. Therefore, in this section, we focus on these two parameterizations and provide a complete dichotomy between tractable and intractable cases.

First, assume that our instance possesses a bounded number of orders~$m$. Then, the topology has at most~$m$ leaves, and therefore, we can directly use the \FPT algorithm from \Cref{thm:MMS:FPT:numLeaves} and efficiently solve even weighted instances.

\begin{corollary}
	When parameterized by the number of orders~$\numOrders$, an MMS and non-wasteful allocation can be found in \FPT time, even if the instance is weighted.
\end{corollary}

A more interesting restriction from both the practical and theoretical perspective is when the number of agents is bounded. Our next result rules out the existence of a polynomial-time algorithm already for instances with two agents and uses a very simple topology. The reduction is from a suitable variant of the \probName{Equitable Partition} problem~\cite{DeligkasEKS2024}.

\begin{theorem}
	\label{thm:MMS:NPh:twoAgents}
	Unless~$\textsf{P}=\NP$, there is no polynomial-time algorithm that finds an MMS and non-wasteful allocation, even if~$G$ is a weighted star and~$|\agents| = 2$.
\end{theorem}
\begin{proof}
	We reduce from a variant of the \probName{Equitable Partition} problem. In \probName{Equitable Partition}, we are given a set~$A=\{a_1,\ldots,a_{2k}\}$ of integers such that~$\sum_{a_i\in A} a_i = 2\cdot B$, and the goal is to decide whether there exists a set~$S\subseteq [2k]$,~$|S| = k$, such that~$\sum_{a\in S} a = \sum_{a\in [2k]\setminus S} a = B$. This problem is known to be weakly \NPh even if for every~$S\subseteq [2N]$ such that~$|S| < k$ it holds that~$\sum_{a\in S} a < B$~\cite{DeligkasEKS2024}.
	
	The construction is very similar to the one used in the proof of \Cref{thm:MMS:NPh:stars}. Again, we create the hub~$h$ and exactly one order~$v_i$ for every element~$a_i\in A$. Each order~$v_i$ is connected to the hub with an edge of weight~$a_i$. To complete the construction, we set~$\agents = \{1,2\}$, and we ask whether the MMS-share of~$\mathcal{I}$ is at most~$q=B$.
	
	Let~$\mathcal{J}$ be a \Yes-instance an let~$S\subseteq [2n]$ be a solution. We create an allocation~$\alloc$ such that~$\alloc_1 = \{ v_i \mid a_i\in S \}$ and~$\alloc_2 = (\orders)\setminus \alloc_1$, and we claim that the maximum cost of a bundle in~$\alloc$ is at most~$q = B$. By the construction of weights and by the fact that~$S$ is a solution for~$\mathcal{J}$, we obtain that it is indeed the case.
	
	In the opposite direction, let~$\mathcal{I}$ be a \Yes-instance and let~$\alloc$ be an allocation such that~$\max\{\cost(\alloc_1),\cost(\alloc_2)\} \leq q$. Without loss of generality, let~$\cost(\alloc_1) \geq \cost(\alloc_2)$. We show that a)~$|\alloc_1| = |\alloc_2|$ and b)~$\cost(\alloc_1) = \cost(\alloc_2)$. For the later property, observe that an allocation that is a social optimal has a total cost of~$2B = 2q$. We assumed that~$\alloc$ is a solution, and therefore, the worst bundle is of the cost at most~$q$. As there are only two bundles, it directly follows that both bundles cost exactly~$q$. Otherwise, some orders would not be serviced, which is not allowed. For the former claim, assume, without loss of generality, that~$|\alloc_1| > |\alloc_2|$. Since there are~$2k$ order, it must hold that~$|\alloc_2| \leq k - 1$. But then, by our assumption about the instance of \probName{Equitable Partition}, it must hold that~$\cost(\alloc_2) < q$. This contradicts the fact that both bundles have the same cost. Now, it remains to take~$S = \{ a_i \mid v_i \in \alloc_1 \}$ and we have a solution for~$\mathcal{J}$, finishing the proof. 
\end{proof}

For unweighted instances, though, \citet[Theorem 5]{HosseiniNW2024} introduced an \XP algorithm capable of finding an MMS allocation. That is, if the instance is unweighted, then for every constant number of agents, there is an algorithm that finds an MMS and non-wasteful allocation in polynomial time. Their result raises the question of whether this parameterization admits a fixed-parameter tractable algorithm. We answer this question negatively by showing that, under the standard theoretical assumptions, \FPT algorithm is not possible, and therefore, the algorithm of \citet{HosseiniNW2024} is basically optimal. Moreover, the topology created in the following hardness proof is so that if we remove a single vertex, we obtain a disjoint union of paths. This time, we reduce from \probName{Unary Bin Packing} parameterized by the number of bins~\cite{JansenKMS2013}.

\begin{theorem}
	\label{thm:MMS:Wh:numAgents}
	Unless~$\FPT=\W$, there is no \FPT algorithm with respect to the number of agents~$|\agents|$ that finds an MMS and non-wasteful allocation, even if the instance is unweighted and the distance to disjoint paths of~$G$ is one.
\end{theorem}
\begin{proof}
	We show the hardness by a parameterized reduction from the \probName{Unary Bin Packing} problem. In this problem, we are given a number of bins~$k$, a capacity of every bin~$B$, and a multi-set of items~$a_1,\ldots,a_\ell$ such that~$\sum_{i=1}^\ell a_i = k\cdot B$. The goal is to find an assignment~$\alpha\colon A\to[k]$ such that for every~$j\in[k]$ it holds that~$\sum_{a_i \in A \mid \alpha(a_i) = j} a_i = B$. This problem is known to be \Wh when parameterized by the number of bins~$k$~\cite{JansenKMS2013}. Without loss of generality, we can assume that for each item~$a_i$ it holds that~$a_i \geq 1$, as zero-valued items can be assigned to an arbitrary bin without changing the solution.
	
	Let~$\mathcal{J} = (k,B,A)$ be an instance of \probName{Unary Bin Packing} problem. We construct an equivalent instance of our problem as follows. For each item~$a_i\in A$, we create an \emph{item-gadget}~$P_i$, which is a simple path with~$a_i$ vertices~$v_i^1,\ldots,v_i^{a_i}$. Next, we add the hub~$h$ and connect it with the vertex~$v_i^1$ of every item-gadget~$P_i$. Finally, the number of agents~$\numAgents$ is~$k$, and we ask whether the MMS-share of this instance is at most~$q=B$.
	
	For correctness, suppose that~$\mathcal{J}$ is a \Yes-instance and~$\alpha$ is a solution assignment. We create an allocation~$\alloc$ such that we set~$\alloc_j = \{  V(X_i) \mid \alpha(a_i) = j \}$ and we claim that~$\max_{j\in\agents} \cost(\alloc_j) \leq q$. For the sake of contradiction, let there be a bundle~$\alloc_j$ with~$\cost(\alloc_j) > q$. By the construction of~$\alloc_j$, it contains all vertices of some item-gadgets~$P_{i_1},\ldots,P_{i_{\ell'}}$. The cost of servicing the whole item-gadget is equal to the number of its vertices. Therefore, it must hold that~$|V(P_{i_1})| + \cdots + |V(P_{i_{\ell'}})| > q$. However, the number of vertices is equal to the corresponding item, and hence, also~$a_{i_1} + \cdots + a_{i_{\ell'}} > q$. However, we constructed~$\alloc_j$ such that all~$a_{i}$ corresponding to~$P_i\in \alloc_j$ are allocated to the same bin in~$\mathcal{J}$. This contradicts that~$\alpha$ is a solution for~$\mathcal{J}$, as~$q = B$. Therefore, there is no bundle whose cost exceeds~$q$, and~$\mathcal{I}$ is also a \Yes-instance.
	
	In the opposite direction, let~$\mathcal{I}$ be a \Yes-instance and~$\alloc$ be an allocation with~$\max_{j\in \agents} \cost(\alloc_j) \leq q$. Without loss of generality, we can assume that~$\alloc$ is a non-wasteful allocation, as otherwise, we use \Cref{thm:NW:easyToTurnInto} to turn it into a non-wasteful one without increasing its maximum cost. By the definition of non-wasteful allocations, now every item-gadget is serviced by exactly one agent. Observe that the size of each bundle must be exactly~$q$, as~$\cost(\orders) = \numAgents\cdot q$, we have~$\numAgents$ agents, and~$\max_{j\in \agents} \cost(\alloc_j) \leq q$. We define an assignment~$\alpha$ such that if the item-gadget~$P_i$ is serviced by an agent~$j$, then we set~$\alpha(a_i) = j$. Now, we argue that~$\alpha$ is a correct solution for~$\mathcal{J}$. For the sake of contradiction, assume the opposite, that is, assume that there exists a bin~$j$ such that~$\sum_{a_i\in A \mid \alpha(a_i) = j} > B$. Then, by the definition of the assignment, the cost of the agent's~$j$ bundle is~$\cost(\alloc_j) > B = q$, which contradicts that~$\alloc$ is a solution. Therefore, such a bin cannot exist, and we showed that~$\mathcal{J}$ is also a \Yes-instance.
	
	Finally, it is easy to see that the reduction can be done in polynomial time. Moreover, we used~$\numAgents = k$, and therefore, the presented reduction is indeed a parameterized reduction, and our problem is therefore \Wh when parameterized by the number of agents~$\numAgents$. Also, observe that if we remove~$h$ from~$G$, we obtain a disjoint union of item-gadgets -- as every item gadget is a path, we get that the distance to disjoint paths of this graph is exactly one. This finishes the proof.
\end{proof}

\section{Restricted Topologies}
\label{sec:topologies}

In this section, we take a closer look at the computational (in)tractability of fair distribution of delivery orders via different restrictions of the topology. Apart from the theoretical significance of such an approach~\cite{IgarashiZ2024,ZhouWLL2024,Schierreich2024}, the study is also driven by a practical appeal. It arises in multiple problems involving maps or city topologies that the underlying graph model usually possesses certain structural properties that can be exploited to design efficient algorithms for problems that are computationally intractable in general (see, e.g., \cite{ElkindPTZ2020,AgarwalEGISV2021,KnopS2023} for a few examples of such studies).

\subsection{Star-Like Topologies}

Topologies isomorphic to \textit{stars} are particularly interesting for applications where, after processing each order, an agent must return to the hub. One such example is moving companies, where loading a vehicle with more than one order at a time is usually physically impossible.

In contrast to the previous intractability for weighted instances, the following result shows that if~$G$ is an unweighted star, then MMS and non-wasteful allocation can be found efficiently.

\begin{proposition}\label{thm:MMS:P:unweightedStar}
	If~$G$ is a star and the input instance is unweighted, an MMS and non-wasteful allocation can be found in linear time.
\end{proposition}
\begin{proof}
	By \Cref{lem:MMS:ifHubIsLeafWeCanReduce}, we can assume that the hub~$h$ is the star's center, as otherwise, we can reduce the instance. Now, under this assumption, we claim that for every instance~$\mathcal{I}$ where the topology~$G$ is an unweighted star, it holds that~$\MMSshare(\mathcal{I}) = \lceil \numOrders / \numAgents \rceil$. Let~$\leaves(G) = \{\ell_1,\ell_2,\ldots,\ell_\numLeaves\}$. Since~$G$ is a star and~$h$ is its center, the cost of servicing a leaf~$\ell_j$ is completely independent of the cost of servicing other leaves. Therefore, for every~$A\subseteq \orders$, we have that~$\cost(A) = |A|$. Now, we create an allocation~$\alloc$ so that for every~$i\in[\numOrders]$, we set~$\ell_i \in \alloc_{i\bmod \numAgents}$. By the construction, the size of the biggest bundle is~$\lceil \numOrders/\numAgents \rceil$, and consequently, the MMS-share of~$\mathcal{I}$ is at most~$\lceil \numOrders/\numAgents \rceil$. To finalize the proof, assume that~$\MMSshare(\mathcal{I}) < \lceil \numOrders/\numAgents \rceil$. Then,~$\sum_{i\in[\numAgents]} \alloc_i < \cost(\orders)$, meaning that there is an order which is not serviced in~$\alloc$. This is impossible, so such an allocation cannot exist, and~$\MMSshare(\mathcal{I}) \geq \lceil \numOrders/\numAgents \rceil$.
\end{proof}

The previous positive results naturally cannot be generalized to the weighted setting as of \Cref{thm:MMS:NPh:twoAgents} already for instances with two agents. However, the hardness in \Cref{thm:MMS:NPh:twoAgents} heavily relies on the fact that the weights of the edges are exponential in the number of orders. This is not a very natural assumption for real-life instances. In practical instances, it is more likely that the weights will be relatively small compared to the number of orders. Fortunately, we show that, for such instances, an efficient algorithm exists for any constant number of agents. The algorithm uses as a subprocedure the \probName{Multi-Way Number Partition} problem, where the goal is to partition a set of numbers~$\mathcal{A}$ into subsets~$A_1,\ldots,A_k$ so that~$\max_{i\in[k]}\sum_{a\in A_i} a$ is minimized. This problem is known to admit a pseudo-polynomial time algorithm~\cite{Korf2009}.

\begin{theorem}
	For every constant~$c\in\N$, if~$G$ is a weighted star and~${|\agents| = c}$, an~MMS and non-wasteful allocation can be found in pseudo-polynomial time.
\end{theorem}
\begin{proof}
	We reduce the problem of finding MMS and non-wasteful allocation of weighted stars to the \probName{Multi-Way Number Partition} problem. In this problem, we are given a set of integers~$A=\{a_1,\ldots,a_N\}$ and the goal is to partition~$A$ into sets~$A_1,\ldots,A_k$ such that the largest sum of elements of a part is minimized. It is known that the \probName{Multi-Way Number Partition} can be solved in~$\Oh{N\cdot (k-1) \cdot (\max A)^{k-1}}$ time (and space) using dynamic-programming~\cite{Korf2009}.
	
	Given an instance~$\mathcal{I}$ of fair distribution of delivery orders, we create an equivalent instance~$\mathcal{J}$ of the \probName{Multi-Way Number Partition} as follows. The set~$A$ contains a single integer~$a_i$ for every order~$\ell_i\in \orders$ and its value is~$\wFn(\{h,\ell_i\})$. To complete the construction, we set~$k = |\agents|$. It is now easy to see that the solution for~$\mathcal{J}$ that minimizes the maximum sum of elements of a part corresponds to an MMS allocation~$\mathcal{I}$, since~$G$ is a star and, therefore, the cost function is additive. Moreover, for stars, every allocation is trivially non-wasteful.
	
	The transformation can be done clearly in linear time. Therefore, given an instance~$\mathcal{I}$ of the fair distribution of delivery orders, we can transform it into an equivalent instance of the \probName{Multi-Way Number Partition}, use the dynamic programming algorithm running in~$\Oh{N\cdot (k-1) \cdot (\max A)^{k-1}}$ to solve~$\mathcal{J}$, and then, again in linear time, reconstruct the solution for~$\mathcal{I}$. Overall, this is an algorithm running in time~$\Oh{\numOrders} + \Oh{N\cdot (k-1) \cdot (\max A)^{k-1}}$, which is~$\Oh{\numOrders\cdot(c-1)\cdot (\max_{\ell\in\orders} \wFn(\{h,\ell\})^{c-1}}$, as~$N=\numOrders$,~$k=c$ and~$\max A = \max_{\ell\in\orders} \wFn(\{h,\ell\})$. That is, for every fixed~$c\in\N$ and the weights encoded in binary, an MMS and non-wasteful allocation can be found in polynomial time.
\end{proof}

\subsection{Bounded-Depth Topologies}

Stars rooted in their center are rather shallow trees; in particular, they are the only family of trees of depth one. It is natural to ask whether the previous algorithms can be generalized to trees of higher depth. In the following result, we show that this is not the case. In fact, our negative result is even stronger and shows that we cannot hope for a tractable algorithm already for unweighted instances of depth two and with diameter four. The reduction is again from the \probName{$3$-Partition} problem.

\begin{theorem}
	\label{thm:MMS:NPh:depthDiameter4pvc}
	Unless~$\textsf{P}=\NP$, there is no polynomial-time algorithm that finds an MMS and non-wasteful allocation, even if the instance is unweighted, the depth of~$G$ is two, the diameter of~$G$ is four, and the~$4$-path vertex cover number of~$G$ is one.
\end{theorem}
\begin{proof}
	We reduce from the \probName{$3$-Partition} problem. Recall that in this problem, we are given a set~$A$ of~$3k$ integers~$a_1,\ldots,a_{3k}$ such that~$\sum_{a\in A} a = k\cdot B$ and our goal is to find a set~$\mathcal{S}$ of~$k$ pairwise disjoint~$3$-sizes subsets of~$A$ such that elements of each~$S\in\mathcal{S}$ sums up to exactly~$B$. Without loss of generality, we can assume that~$B/4 < a < B/2$ for every~$a\in A$.
	
	Given an instance~$\mathcal{J}$ of \probName{$3$-Partition}, we construct an equivalent instance of our problem as follows (see \Cref{fig:MMS:NPh:depthDiameter4pvc:construction} for an illustration of the construction). First, for every element~$a_i\in A$, we create a \emph{element-gadget}~$X_i$, which is a star with~$a_i$ leaves and a center~$c_i$. The centers of all element-gadgets are connected by an edge with the hub~$h$. The instance is unweighted, and we have~$\numAgents=k$ agents. The question is whether the MMS-share of this instance is at most~$q=B+3$.
	
	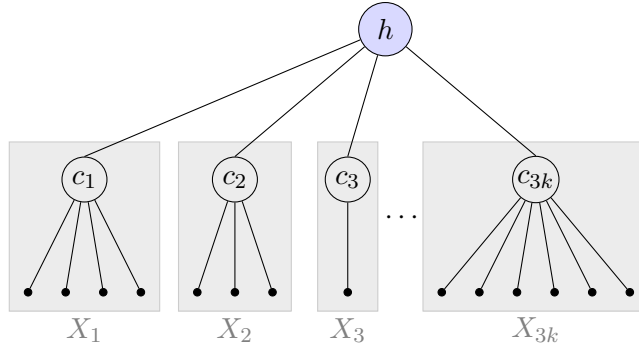
\begin{figure}
		\centering
		\begin{tikzpicture}
			\node[draw,circle,fill=blue!15] (h) at (0,1) {$h$};
			
			\node[draw,circle,inner sep=2pt] (c1) at (-4,-1) {$c_1$};
			\foreach[count=\i] \x in {-4.75,-4.25,...,-3.25} {
				\node[draw,fill=black,circle,inner sep=1pt] (c1\i) at (\x,-2.5) {};
				\draw (c1) -- (c1\i);
			}
			\begin{pgfonlayer}{background}
				\draw[gray!40,fill=gray!15] (-5,-0.5) rectangle (-3,-2.75);
				\node[gray] at (-4,-3) {$X_1$};
			\end{pgfonlayer}
			
			\node[draw,circle,inner sep=2pt] (c2) at (-2,-1) {$c_2$};
			\foreach[count=\i] \x in {-2.5,-2,...,-1.5} {
				\node[draw,fill=black,circle,inner sep=1pt] (c2\i) at (\x,-2.5) {};
				\draw (c2) -- (c2\i);
			}
			\begin{pgfonlayer}{background}
				\draw[gray!40,fill=gray!15] (-2.75,-0.5) rectangle (-1.25,-2.75);
				\node[gray] at (-2,-3) {$X_2$};
			\end{pgfonlayer}
			
			\node[draw,circle,inner sep=2pt] (c3) at (-0.5,-1) {$c_3$};
			\foreach[count=\i] \x in {-0.5,...,-0.5} {
				\node[draw,fill=black,circle,inner sep=1pt] (c3\i) at (\x,-2.5) {};
				\draw (c3) -- (c3\i);
			}
			\begin{pgfonlayer}{background}
				\draw[gray!40,fill=gray!15] (-0.9,-0.5) rectangle (-0.1,-2.75);
				\node[gray] at (-0.5,-3) {$X_3$};
			\end{pgfonlayer}

			\node (dummy) at (0.25,-1.5) {$\cdots$};
			
			\node[draw,circle,inner sep=0.75pt] (ck) at (2,-1) {$c_{3k}$};
			\foreach[count=\i] \x in {0.75,1.25,1.75,...,3.25} {
				\node[draw,fill=black,circle,inner sep=1pt] (ck\i) at (\x,-2.5) {};
				\draw (ck) -- (ck\i);
			}
			\begin{pgfonlayer}{background}
				\draw[gray!40,fill=gray!15] (0.5,-0.5) rectangle (3.5,-2.75);
				\node[gray] at (2,-3) {$X_{3k}$};
			\end{pgfonlayer}
			
			\draw (h) 
			edge (c1.north) edge (c2.north) edge (c3.north)
			edge (ck.north);
		\end{tikzpicture}
		\caption{An illustration of the construction used to prove \Cref{thm:MMS:NPh:depthDiameter4pvc}.}
		\label{fig:MMS:NPh:depthDiameter4pvc:construction}
	\end{figure}
	
	Let~$\mathcal{J}$ be a \Yes-instance and let~$\mathcal{S} = \{S_1,\ldots,S_k\}$ be  solution. For every agent~$i\in[\numAgents]$, we set~$\alloc_i = \bigcup_{j\mid a_j\in S_i} V(X_j)$, that is, every agent services all orders of an element-gadget corresponding to an element which is present in~$S_i$ in the solution~$\mathcal{S}$. Now, let~$i\in [k]$ be an arbitrary agent. Then its cost is~$\cost(\alloc_i) = \cost(X_{i_1}) + \cost(X_{i_2}) + \cost(X_{i_2})$, as different element-gadgets are clearly independent. This can be rewritten as~$\cost(\alloc_i) = 1 + |X_{i_1}\setminus\{c_{i_1}\}| + 1 + |X_{i_2}\setminus\{c_{i_2}\}|+ 1 + |X_{i_3}\setminus\{c_{i_3}\}|$, where~$+ 1$ is for servicing the center of an element-gadget and~$|X_{i_\ell}\setminus\{c_{i_\ell}\}|$ is the minimum required contribution to the cost for visiting the leaves of an element-gadget~$X_{i_\ell}$, assuming that the agent already services the center of this element-gadget. Since~$\sum_{a\in S} a = B$ for every~$S\in \mathcal{S}$, then by the construction, it holds that~$\cost(\alloc_i) = B + 3$ for every~$i\in[\numAgents]$. Therefore, the MMS-share of~$\mathcal{I}$ is at most~$B+3$.
	
	In the opposite direction, let~$\mathcal{I}$ be an instance of MMS-share at most~$q$ and let~$\alloc$ be an allocation such that~${\max_{i\in[\numAgents] \cost(\alloc_i)} \leq q}$. First, assume that there exists at least one bundle such that~${\cost(\alloc_i) < q}$. Then,~$\sum_{i=1}^\numAgents \cost(\alloc_i) \leq \numAgents\cdot q - 1 = \numAgents\cdot(B+3) - 1 = \numAgents\cdot B + 3\numAgents - 1$; however, there are~$\numAgents\cdot B + 3\numAgents$ orders, meaning that at least one order is not services, which contradicts that~$\alloc$ is a solution. Hence, the cost of each bundle is exactly~$q$, and also, the MMS-share of~$\mathcal{I}$ is exactly~$q$. Now, we show that all orders of the same element-gadget are serviced by the same agent. 
	
	\begin{claim}
		For each~$X_j$,~$j\in[3k]$, there exists exactly one~$i\in\agents$ such that~$V(X_j) \subseteq \alloc_i$.
	\end{claim}
	\begin{claimproof}
		We know that in every MMS allocation, the costs of all bundles are exactly~$q$. Since there are~$\numAgents$ bundles, the total cost is~$\numAgents\cdot q = \numAgents\cdot (B+3) = \numAgents\cdot B + 3\numAgents$, which is also the total number of orders. Therefore, each edge is visited exactly twice. However, if an element-gadget, say~$X_j$, would be visited by two agents, the edge~$\{h,c_j\}$ would be traversed at least four times. Hence, the sum of costs is at least~$\numAgents\cdot q + 2$, which contradicts that the cost of each bundle is exactly~$q$.
	\end{claimproof}
	
	Finally, assume that there is an agent servicing more than three element-gadgets. Recall that each element-gadget has strictly more than~$q/4$ leaves. Hence, the cost of the bundle of such an agent would be~$(q/4 + 1)\cdot 4 = q + 4$, which contradicts that~$\MMSshare(\mathcal{I}) = q$. Hence, every agent services at most three element-gadgets. If there would be an agent servicing less than three element-gadgets, then by simple counting argument, we obtain that there is some other agent servicing more than three element-gadgets, which is again a contradiction. Therefore, each agent services exactly three element-gadgets. Now, we create a solution~$\mathcal{S} = (S_1,\ldots,S_k)$ such that we set~$S_i = \{ a_j \mid V(X_j) \subseteq \alloc_i \}$. By previous argumentation, each set is obviously of size three. If the sum of elements of some set~$S_i\in \mathcal{S}$ is not~$B$, then we have that the cost of corresponding~$\alloc_i$ is not~$q$, which is a contradiction. Therefore, the sum of each~$S_i$ is exactly~$B$, and~$\mathcal{J}$ is also a \Yes-instance.
	
	It is easy to see that the depth of the constructed graph~$G$ is always two and that its diameter is four. To see that the~$4$-path vertex cover number of~$G$ is one, observe that if we remove~$h$ from~$G$, we obtain a disjoint union of stars. The longest path in a star is of length two, which shows that~$\{h\}$ is indeed a modulator to graphs without a path of length three.
\end{proof}

The structural parameter~$4$-path vertex cover mentioned in the previous result can be seen as the minimum number of vertices we need to remove from the topology to obtain a disjoint union of stars. That is, topologies with bounded~$4$-path vertex cover are generalizations of stars and apply to an even wider variety of real-life instances.

In contrast to the previous hardness result, we show that if the problem is parameterized by the~$3$-path vertex cover number of the topology, there exists an \FPT algorithm. A set of vertices~$C$ is called the~$3$-path vertex cover ($3$-PVC) if the graph~$G' = (V\setminus C, E)$ is a graph of maximum degree one. The size of the smallest possible~$3$-PVC is then called the~$3$-path vertex cover number or dissociation number of~$G$~\cite{PapadimitriouY1982}. 
This parameter, albeit less common, has been used to obtain tractable algorithms in several areas of artificial intelligence and multiagent systems~\cite{XiaoLDZ2017,GruttemeierKM2021,KnopSS2022,GruttemeierK2022}, and is also a generalization of the well-known \textit{vertex cover}; if we remove vertex cover vertices, we obtain a graph of maximum degree zero.
It is worth mentioning that a minimum size~$3$-PVC of a tree can be found in polynomial time~\cite{PapadimitriouY1982}. Therefore, any algorithm for the fair division of delivery orders can first check whether the topology possesses bounded~$3$-PVC and, if yes, employ our algorithm.

\begin{theorem}
	\label{thm:MMS:FPT:3pvc}
	If the instance is unweighted and parameterized by the~$3$-path vertex cover number~$\vartheta$ and the number of different weights~$\psi$, combined, an MMS and non-wasteful allocation can be found in \FPT time.
\end{theorem}
\begin{proof}
	It is known that if~$G$ is a tree, there always exists an optimal~$3$-path vertex cover~$C$ such that it contains no leaf of~$G$~\cite{PapadimitriouY1982}. Therefore, let~$C$ be such a cover and let~$|C|=\vartheta$. Also, by \Cref{lem:MMS:ifHubIsLeafWeCanReduce}, we may assume that~$h$ is not a leaf. First, we show that in~$V\setminus C$, there are less than~$\vartheta$ vertices with both the parent and at least one child in~$C$. For the sake of contradiction, assume that there exists at least~$\vartheta$ vertices~$X = x_1,\ldots,x_\vartheta$ such that each of them has both the parent and at least one child in~$C$. As each vertex~$X$ must have a unique child, and there must be at least one parent vertex, we have that~$C$ is of size at least~$\vartheta + 1$. This contradicts that~$C$ is of size~$\vartheta$. Consequently,~$X$ is of size at most~$\vartheta - 1$. We add all vertices of~$X$ into~$C$, which increases the size of~$C$ at most twice.
	
	According to the previous argumentation, all vertices outside~$C$ are either leaves or vertices of degree two with their parent in~$C$ (recall that leaves are never part of~$C$). Let~$v$ be a degree two vertex outside of the cover. We remove~$v$ from the graph, add an edge connecting the only child~$u$ of~$v$ with the parent~$p$ of~$v$, and set~$\wFn(\{u,p\}) = \wFn(\{u,v\}) + \wFn(\{v,p\})$. That is, we turn the instance into the weighted one; however, the weight function uses at most~$\psi+\psi^2 \in \Oh{\psi^2}$ different values. Moreover, after this operation, all vertices outside of~$C$ are leaves, so the number of internal vertices is at most~$2\cdot\vartheta$. Consequently, we can use the algorithm of \Cref{thm:MMS:FPT:numInternal} to solve the reduced instance in \FPT time, finishing the proof.
\end{proof}

The algorithm from \Cref{thm:MMS:FPT:3pvc} uses as the sub-procedure the \FPT algorithm for the parameterization by the number of internal vertices and the number of different edge-weights. In fact, we show that any instance with~$3$-pvc~$\vartheta$ and~$\psi$ different edge-weights can be transformed to an equivalent instance with~$\Oh{2\vartheta}$ internal vertices and~$\Oh{\psi^2}$ different edge-weights. Such a reduced instance can then be directly solved in \FPT time by the algorithm from \Cref{thm:MMS:FPT:numInternal}.

\subsection{Topologies with a Central Path}

When the topology is a simple path, we can find an MMS and non-wasteful allocation in polynomial time: just allocate each leaf to a different agent. Moreover, this approach works even if the instance is weighted.

\begin{proposition}
	If~$G$ is a path, an MMS and non-wasteful allocation can be found in linear time, even if the instance is weighted. 
\end{proposition}
\begin{proof}
	By \Cref{lem:MMS:ifHubIsLeafWeCanReduce}, we know that the hub is not a leaf. Therefore, there are two leaves,~$\ell_1$ and~$\ell_2$. Without loss of generality, assume that the shortest path from~$h$ to~$\ell_2$ is at most as long as the shortest path from~$h$ to~$\ell_1$. We define a solution allocation so that~$\alloc_1 = W_{\ell_1}\setminus\{h\}$,~$\alloc_2 = W_{\ell_2}\setminus\{h\}$, and~$\alloc_i = \emptyset$ for every~$i\in[3,\numAgents]$. Since every bundle contains all vertices on a path from~$h$ to the respective leaf, the allocation~$\alloc$ is clearly non-wasteful. It is also easy to see that the allocation is MMS. The lower-bound on the MMS-share is the distance to the most distant leaf from~$h$, and our allocation achieves this bound. 
\end{proof}

Therefore, the following set of results explores the complexity picture for instances that are not far from being paths. More specifically, we focus on topologies where all vertices are at a limited distance from a \emph{central path}. Such topologies may appear in practice very naturally, e.g., in instances where the central path is a highway, and the other vertices represent smaller towns along this highway.

Unfortunately, by the intractability results for weighted stars (cf. \Cref{thm:MMS:NPh:stars}), we cannot expect any tractable algorithms for topologies with distance to the central path greater or equal to one. 
Nonetheless, focusing on unweighted instances, we give a polynomial time algorithm for graphs where each vertex is at a distance at most one from the central path; such graphs are commonly known as \emph{caterpillar trees}.

\begin{theorem}
	\label{thm:MMS:P:unweightedCaterpillars}
	If~$G$ is a caterpillar tree and the instance is unweighted, an MMS and non-wasteful allocation can be found in polynomial time.
\end{theorem}
\begin{proof}
	Before we proceed to the description of the algorithm, recall that we can assume that~$h$ is not a leaf as if it would be the case, we can reduce to an equivalent instance where~$h$ is an internal vertex by \Cref{lem:MMS:ifHubIsLeafWeCanReduce}. Additionally, since~$G$ is a caterpillar graph, the hub~$h$ has at most two neighbors of degree greater than one. If all neighbors of~$h$ are leaves, the graph is a star, and we can use \Cref{thm:MMS:P:unweightedStar} to solve such an instance in polynomial time. Therefore, there is at least one neighbor, call it~$r$, of~$h$ with a degree of at least two. There also may exist a second neighbor with a degree of at least two; we use~$\ell$ to refer to this neighbor. We define two sub-trees~$T_r$ and~$T_\ell$, where~$T_\ell$ is (possibly empty) sub-tree rooted in the vertex~$\ell$, and~$T_r$ is~$G\setminus T_\ell$.
	
	First, we assume that~$T_\ell$ is an empty tree; that is, the hub~$h$ is an endpoint of the central path. We denote the central path as~$(v_1,\ldots,v_x = h)$ and for every~$i\in[x]$, we use~$L_i$ to denote the (possibly empty) set of leaves whose parent is vertex~$v_i$. As the final step, we guess the MMS-share~$q$ of the input instance~$\mathcal{I}$. The core of the algorithm is then a sub-procedure which, given an instance~$\mathcal{I}$ and a value~$q\in\N$, finds an allocation~$\alloc$ such that~$\max_{i\in\agents} \cost(\alloc_i) \leq q$, if such~$\alloc$ exists, or decides that such an allocation cannot exist. Since the instance is unweighted, the MMS-share of each input instance lies in the interval~$[\max_{v\in\orders} \dist(v,h),\numOrders]$. This implies that there are only linear-many candidate values~$q$ for the MMS-share, and we can invoke the sub-procedure with each of them in increasing order and stop when the sub-procedure returns an allocation for the first time.
	
	The sub-procedure works greedily. First, it allocates as many leaves of~$v_1$ as possible to the agent~$\numAgents$. If there are more orders in~$L_1$ than agent~$\numAgents$ can service to keep the cost of~$\alloc$ at most~$q$, we add to~$\alloc_\numAgents$ as many leaves of~$v_1$ as possible, remove these leaves from~$G$, reduce the number of agents by one, and run the same procedure with the reduced instance. If we can add all leaves of~$L_1$ to~$\alloc_\numAgents$, we also include~$v_1$ in~$\alloc_\numAgents$, remove~$L_1 \cup \{v_1\}$ from~$G$, and rerun the same operation with reduced~$G$ and partially filled~$\alloc_\numAgents$. The procedure terminates if either all orders are allocated or the cost of all bundles is exactly~$q$. In the first case, the sub-procedure verified that an allocation with the maximum cost of~$q$ exists; in the latter case, there is no way how to partition orders between agents so that the maximum cost of a bundle is~$q$ and the sub-procedure returns \No. Formally, the sub-procedure works as is in \Cref{alg:caterpillars}.
	
	\begin{algorithm}[tb]
		\caption{An algorithm for deciding whether an allocation~$\alloc$ with~$\max_{i\in\agents} \cost(\alloc_i) \leq q$ exists.}
		\label{alg:caterpillars}
		\textbf{Input}: A problem instance~$\mathcal{I}=(G,h,\agents)$, where~$G$ is a caterpillar graph with central path~$(v_1,\ldots,v_x = h)$, an integer~$q\in\N$.\\
		\textbf{Output}: An allocation~$\alloc$ such that~$\max_{i\in[\agents]} \cost(\alloc_i) = q$, if one exists,~$\No$ otherwise.
		\begin{algorithmic}[1] %
			\State~$j \gets 1$ \Comment{A pointer to the currently processed vertex of the central path.}
			\For{$i\gets \numAgents,\ldots,1$}
			\State~$\alloc_i = \emptyset$
			\State~$C \gets \dist(v_j,h)$ \Comment{The current cost of~$\alloc_i$.}
			\While{$C < q$ \textbf{and}~$j \leq x$}
			\State~$s \gets q - C$ \Comment{Compute the remaining capacity of~$\alloc_i$.}
			\If{$|L_j| > s$} \Comment{Case 1: Agent~$i$ cannot service the whole~$L_j$.}
			\State~$\alloc_i \gets \alloc_i \cup \{ l_j^y \mid y \in [|L_j|- s + 1,|L_j|] \}$ \Comment{Add as many leaves of~$L_j$ as possible...}\label{alg:caterpillars:Case1}
			\State~$L_j \gets L_j \setminus \alloc_i$ \Comment{...and remove them from~$G$.}
			\State~$C \gets C + s = q$
			\Else \Comment{Case 2: Agent~$i$ can service the whole~$L_j$.}
			\State~$\alloc_i \gets \alloc_i \cup L_j \cup \{v_j\}$\label{alg:caterpillars:Case2}
			\State~$C \gets C + |L_j|$
			\State~$j \gets j + 1$
			\EndIf
			\EndWhile
			\EndFor
			\If{$\bigcup_i^\numAgents \alloc_i = V$} \Comment{If all orders allocated, we have a solution.} \label{alg:caterpillars:isAlloc}
			\State \Return~$\alloc$
			\Else
			\State \Return \No
			\EndIf
		\end{algorithmic}
	\end{algorithm}
	
	Overall, the described algorithm is clearly polynomial: \Cref{alg:caterpillars} consists of two loops, the outer loop is limited to~$\Oh{\numAgents}$ executions and the inner one to~$\Oh{m}$ executions, and we invoke this algorithm~$\Oh{m}$ times. Now, we show the correctness of the algorithm using several auxiliary claims.
	
	First, we show that if an algorithm outputs a positive result, it is indeed a valid allocation.
	
	\begin{claim}\label{lem:MMS:P:unweightedCaterpillars:isAlloc}
		If \Cref{alg:caterpillars} outputs some~$\alloc$, then it holds that~$\bigcup_i^\numAgents \alloc_i = V$ and~$\alloc_i \cap \alloc_j = \emptyset$ for each pair of distinct~$i,j\in[\numAgents]$.
	\end{claim}
	\begin{claimproof}
		First, assume that there is an unallocated order~$v\in\orders$. Then, the condition on \Cref{alg:caterpillars:isAlloc} is not met, so the algorithm outputs \No. Hence, such a situation can never happen. Now, assume that an order~$v\in\orders$ is part of multiple bundles. We distinguish two cases according to whether~$v$ is a leaf or an internal vertex. An internal vertex~$v=v_j$ is added to a bundle~$\alloc_i$ only on \cref{alg:caterpillars:Case2}. However, once this is done, we increase the value of pointer~$j$ by one, and since the value of the pointer never decreases again, this is the last step in which the algorithm assumes vertex~$v_j$. Consequently, no internal vertex may be part of multiple bundles. Let~$v$ be a leaf such that~$v\in L_j$ for some~$j\in[x]$. First, assume that the first time the leaf~$v$ was allocated was on \cref{alg:caterpillars:Case2}. Then, by the same argumentation as with internal vertices, the set of vertices~$L_j$ is never assumed again by the algorithm, so~$v$ is necessarily part of exactly one bundle. Finally, assume that the first time the leaf~$v$ is allocated to a bundle is on \cref{alg:caterpillars:Case1}. Then, such a vertex is immediately removed from~$L_j$, so it can never be allocated once again by the algorithm. Thus,~$\alloc$ is indeed a partition of~$V$.
	\end{claimproof}
	
	Next, we show that the maximum cost of a bundle in the allocation returned by \Cref{alg:caterpillars} never exceeds the given bound~$q$ on the bundle cost.
	
	\begin{claim}\label{lem:MMS:P:unweightedCaterpillars:hasCorrectCost}
		If \Cref{alg:caterpillars} outputs some~$\alloc$, then it holds that~$\max_{i\in\agents} \cost(\alloc_i) \leq q$. 
	\end{claim}
	\begin{claimproof}
		For the sake of contradiction, assume that \Cref{alg:caterpillars} returned an allocation such that~$\cost(\alloc_i) > q$ for some~$i\in[\agents]$. Let~$j\in[x]$ be the first round when the cost of~$\alloc_i$ exceeded~$q$. Observe that, due to the inner loop condition, this is necessarily the last execution of the inner \textbf{while} loop for agent~$i$. We distinguish two cases. First, assume that the condition on \cref{alg:caterpillars:Case1} is satisfied, i.e., the number of leaves of~$L_j$ is strictly greater than~$q - \cost(\alloc_i)$. Then, the algorithm adds to~$\alloc$ exactly~$q-\cost(\alloc_i)$ leaves of~$L_j$. However, each added leaf increases the cost of~$\alloc_i$ by exactly one, as agent~$i$ anyway visits~$v_j$ when servicing~$\alloc_i$. Therefore, the condition on \cref{alg:caterpillars:Case1} is not met, so the algorithm necessarily performed the else block starting on \cref{alg:caterpillars:Case2}. Here, we extend~$\alloc_i$ with~$L_j$ and~$v_j$. We know that~$|L_j| \leq q - \cost(\alloc_i)$ and thus,~$\cost(\alloc_i \cup L_j) \leq q$ by the previous argumentation that~$i$ visits~$v_j$ on the shortest walk needed to service~$\alloc_i$. Moreover, due to this, adding~$v_j$ does not increase the cost of~$\alloc_i$. Overall, we obtain that the algorithm never produces a bundle of cost greater than~$q$, finishing the proof.
	\end{claimproof}
	
	\Cref{lem:MMS:P:unweightedCaterpillars:isAlloc,lem:MMS:P:unweightedCaterpillars:hasCorrectCost} together show that whenever \Cref{alg:caterpillars} produces a positive outcome, it is indeed correct. Now, we show the opposite direction; whenever an allocation~$\alloc'$ with~$\max_{i\in\agents}\cost(\alloc'_i) \leq q$ exists for some~$q$, the algorithm returns a positive outcome. To show this, we use the following claim about the structure of such allocations.
	
	\begin{claim}\label{lem:MMS:P:unweightedCaterpillars:structure}
		Let~$\mathcal{I}$ be an instance and~$q\in\N$ be an integer. If there exists an allocation~$\alloc$ such that~$\max_{i\in\agents} \cost(\alloc_i) \leq q$, then there exists an allocation~$\alloc'$ with~$\max_{i\in\agents} \cost(\alloc'_i) \leq q$ such that for each agent~$i\in\agents$, there is an interval~$J_i = \{j_i,j_i + 1,\ldots,j_i+y\}$ with~$j_i \geq 1$ and~$y \leq x$, such that~$\alloc'_i \subseteq \bigcup_{j=j_i}^{j_i + y} L_j \cup \{v_j\}$ and~$L_j \cap \alloc'_i = L_j$ for every~$j\in[j_i + 1,J_i + y -1]$.
	\end{claim}
	\begin{claimproof}
		If~$\alloc$ is of the correct structure, we are done. Hence, let an agent exist for which the condition is not satisfied. From all agents for which an interval~$J_i$ does not exist, we take an agent~$i\in\agents$ such that it minimizes~$j$ for which~$L_j\cap \alloc_i \not= \emptyset$. Let~$j$ be the minimal index such that~$L_j\cap\alloc_i \not=\emptyset$,~$j'$ the maximal such index, and~$J'_i = \{j,j+1,\ldots,j+y' = j'\}$ be an interval. Since the condition from the claim statement is not satisfied,~$\alloc_i$ either contains an internal vertex~$v\notin \{ v_j, \ldots, v_{j'} \}$, or~$L_{j^*}\cap\alloc_i\not= L_{j^*}$ for some~$j^*\in J'_i\setminus\{j,j'\}$. If the first case occurs, then we reallocate the internal vertex~$v_{j^*}$,~$j^*\not\in J'_i$. to an arbitrary agent~$i'$ such that~$\alloc_{i'} \cap L_{j^*} \not= \emptyset$. As~$i'$ is servicing some leaf in~$L_{j^*}$, the transfer of~$v_{j^*}$ does not increase the cost of~$\alloc_{i'}$; the cost of~$\alloc_i$ remains the same or decreases. That is, such an operation is safe, and we can perform it repeatedly as long as such an internal vertex exists. Hence, we can assume that such an internal vertex does not exist in~$\alloc'_i$, so there must be~$j^*\in J'_i\setminus\{j,j'\}$ such that~$L_{j^*}\cap\alloc_i \not= L_{j^*}$. Let~$j^*$ be the minimum index that breaks this condition. Since~$L_{j^*}\cap\alloc_i \not= L_{j^*}$, there exists an agent~$i'$ servicing at least one order of~$L_{j^*}$. We keep swapping an order~$o'\in L_{j^*} \cap \alloc_{i'}$ with~$o\in L_{j'} \cap \alloc_i$ as long as both orders~$o'$ and~$o$ exist. It is easy to see that none of these swaps increases the cost for either~$i$ or~$i'$, since~$i$ anyway visits~$v_{j^*}$ and~$i'$ anyway visits~$v_{j'}$ and the instance is unweighted. If we end up in the situation when~$L_{j^*} \cap \alloc_{i'} = \emptyset$, we add~$v_{j^*}$ to~$\alloc_i$ and do the same with another agent~$i'$ such that it services at least one leaf of~$L_{j^*}$. If no such agent exists, we have~$\alloc_i \cap L_{j^*} = L_{j^*}$ and continue with another possible~$j^*$. Similarly, if no such~$j^*$ exists, the agent's~$i$ bundle can be described by an interval~$J_i$ and is no longer problematic. On the other hand, if we first exhaust~$L_{j'}$, we set~$J'_{i} = \{j,\ldots,j'-1\}$ and continue with the same modifications. To conclude, we showed that the bundle of each agent~$i$ not satisfying the condition from the claim statement can be turned into a bundle satisfying this condition. Moreover, all the modifications required can be performed in polynomial time, and once an agent is fixed by the algorithm, it is never assumed again; this follows from the property that the algorithm always assumes an agent whose bundle separates the graph~$G$ into two parts: a `left' side where all agents already satisfy the condition, and potentially broken `right' side.
	\end{claimproof}
	
	\Cref{lem:MMS:P:unweightedCaterpillars:structure} guarantees that there always exists an allocation with a structure that is almost the same as the one checked by \Cref{alg:caterpillars}. Such allocations and the allocation constructed in \Cref{alg:caterpillars} differ only in two things: a) in the ordering of bundles and b) in the allocation of internal vertices. We can fix a) by appropriate permutation of bundles, which is always possible as the cost functions of our agents are identical, and b) by allocating an internal vertex~$v_j$ to the first agent~$i$ such that~$\alloc_i\cap L_i \not= \emptyset$. This operation does not affect the cost of~$\alloc_i$, as~$i$ necessarily traverses~$v_i$ in its shortest walk. Hence, the following corollary indeed holds.
	
	\begin{corollary}
		For every~$q\in\N$ and every~$\mathcal{I}$, \Cref{alg:caterpillars} returns an allocation~$\alloc$ with~$\max_{i\in\agents} \cost(\alloc_i)$ if and only if such an allocation exists.
	\end{corollary}
	
	By this, \Cref{alg:caterpillars} is correct. That is if we know the MMS-share~$q$ of our instance, \Cref{alg:algorithm} finds a corresponding MMS allocation in polynomial time. As we try all possible values of~$q$ in increasing order and return the first allocation returned, the overall approach is trivially correct.
	
	It remains to show how to extend the previous approach from instances with the central path of the form~$(v_1,\ldots,v_x = h)$ to instances where the hub~$h$ is an arbitrary vertex~$v_{j^*}$ of the path. In this case, we first the bundle of agent~$\numAgents$; we want this bundle to be in line with \Cref{lem:MMS:P:unweightedCaterpillars:structure}, so it can be described by four numbers~$j,j',n_j,n_{j'}$ with the following meaning: the bundle~$\alloc_n$ can be described by an interval~$J_n = \{j,\ldots,j'\}$ with~$j^*\in J_n$,~$|\alloc_n \cap L_j| n_j$, and~$|\alloc_n \cap L_{j'}| n_{j'}$. For each such quadrupled, we remove the allocated leaves from~$G$, and for all~$n_r$ and~$n_\ell$ such that~$n_r + n_\ell = \numAgents - 1$ we find, using the previous approach, minimum~$q$ such that both~$T_r$ and~$T_\ell$ admits an allocation with the maximum bundle cost of~$q$ using~$n_r$ and~$n_\ell$ agents, respectively.
\end{proof}

The natural subsequent question is whether we can generalize the algorithm from the previous section to larger distances from the central path. It turns out that, without further restriction, this is not the case. In fact, the topology used in the proof of \Cref{thm:MMS:NPh:depthDiameter4pvc} has all vertices at a distance at most two from the central path, and the created instance is unweighted.

\begin{corollary}
	Unless~$\textsf{P}=\NP$, there is no polynomial time algorithm that finds an MMS and non-wasteful allocation, even if all vertices are at a distance at most two from the central path, the central path consists of a single vertex, and the instance is unweighted.
\end{corollary}

\section{Envy-Based Fairness}

In the previous section, we focused on the algorithmic aspects of computing MMS and non-wasteful allocations. Here, we change the perspective and study how much the complexity picture changes if we assume different fairness notions based on a pair-wise comparison of agents' bundles.

\subsection{Envy-Freeness}

We start with the well-known envy-freeness. Recall that for an allocation~$\alloc$ to be envy-free, it must hold that the cost of every bundle is the same. Not only that envy-free allocations are not guaranteed to exist, but in our first result, we show that envy-freeness is inherently incompatible with the notion of non-wastefulness.

\begin{proposition}\label{thm:EF:incompatibleWithNonwasteful}
	There is an instance~$\mathcal{I}$ of fair distribution of delivery orders such that the sets of envy-free allocations and non-wasteful allocations are both non-empty, but their intersection is empty.
\end{proposition}
\begin{proof}
	Let the topology be a path with four vertices and~$u$ and~$v$ be leaves of~$G$. We place the hub on the single neighbor of~$u$. There are two possible envy-free allocations, and in any of them, one bundle contains only the vertex~$v$, while the other bundle contains all the remaining orders. This allocation is clearly not non-wasteful as non-wastefulness requires the neighbor of~$v$ to be serviced by the same agent that is servicing the order~$v$. Similarly, any non-wasteful allocation allocates both~$v$ and the neighbor of~$v$ to the same agent. Such an allocation is, however, not envy-free, as the cost of the other bundle is at most one, which is strictly smaller than the cost of the bundle containing~$v$.
\end{proof}

Now, we focus on the computational complexity side of deciding whether an EF and non-wasteful allocation exists. In every hardness reduction we provide for MMS and non-wastefulness in this paper, we construct the instance so that it is a \Yes-instance if and only if all bundles are of the same cost. Consequently, every hardness result we show for the computation of the MMS and non-wasteful allocations also applies to the \NP- or \Wh{}ness of deciding whether an EF and non-wasteful allocation exists. This is clearly not the case for our algorithmic results.

\subsection{Envy-Freeness up to One Order}

As argued in the previous subsection, the direction of envy-freeness does not promise many positive results. Therefore, we now focus on a relaxation of envy-freeness called envy-freeness up to one order. Here, we allow for a slight difference between the costs of agent bundles; specifically, we require that each envy is eliminated by removing one order from the bundle of the envious agent.

In the case of envy-freeness, we discussed how the hardness results for MMS carries over. For EF1, this is not the case, at least not as directly as for EF. However, even this relaxation is very incompatible with non-wastefulness, as we show in the following result. The arguments are almost the same as in the case of EF; we just need to extend the topology a little.

\begin{proposition}\label{thm:EF1:incompatibleWithNonwasteful}
	There is an instance~$\mathcal{I}$ of fair distribution of delivery orders such that the sets of EF1 allocations and non-wasteful allocations are both non-empty, but their intersection is empty.
\end{proposition}
\begin{proof}
	We extend the topology from \Cref{thm:EF:incompatibleWithNonwasteful} by attaching one order~$w$ to the vertex~$v$. Now, every non-wasteful allocation must allocate~$w$,~$v$, and the neighbor of~$v$ to a single agent~$i$. The bundle of agent~$j\not= i$ then services only one order, so the agent~$i$ is envious towards~$j$. This envy cannot be removed even if we remove the order~$w$ from~$\alloc_i$. On the other hand, in any EF1 allocation, the orders~$v$,~$w$, and the neighbor of~$v$ must be in different bundles, which directly rules out non-wastefulness.
\end{proof}

It is known that both in the standard fair division of indivisible items and the fair division of delivery orders~\cite{HosseiniNW2024}, an EF1 allocation can be found by a (variant of) standard algorithms, such as round-robin~\cite{CaragiannisKMPS2019} or envy-cycle elimination~\cite{LiptonMMS04}. Our first result shows that these algorithms fail to find EF1 and non-wasteful allocations, even if such allocations exist.

\begin{proposition}\label{thm:EF1:standardAlgorithmsDoNotWorkForNonWastefulness}
	Both round-robin and envy-cyle elimination mechanisms fail to find an EF1 and non-wasteful allocation, even though such allocation exists.
\end{proposition}
\begin{proof}
	Let the instance contain two agents, and the graph be as depicted in \Cref{fig:EF1:standardAlgorithmsDoNotWorkForNonWastefulnessInstance}. In every non-wasteful allocation, we are only interested in the allocation of leaves, as the particular allocation of internal vertices does not affect the cost of these agents. Clearly, there is an EF1 and non-wasteful allocation; if one agent service leaves~$\ell_1$ and~$\ell_2$ together with all orders of~$W_{\{\ell_1\}}$ and~$W_{\{\ell_2\}}$ and the other agents services the rest of the orders, the cost of both bundles is exactly five. That is such an allocation is even envy-free. Let us now show that both round-robin and envy-cycle elimination mechanisms fail to find any EF1 and non-wasteful allocation. We show it separately for each mechanism.
	
	\begin{figure}
		\centering
		\begin{tikzpicture}
			\tikzstyle{A1} = [draw,circle,ultra thick,green!80!black,inner sep=6pt,fill=none]
			\tikzstyle{A2} = [draw,ultra thick,red!30,inner sep=8pt,fill=none]
			
			\node[draw,circle,fill=blue!15] (h) at (0,0) {$h$};
			\node[draw,circle] (v1) at (0,-1.5) {};
			\node[A1] at (0,-1.5) {};
			\node[draw,circle] (v2) at (-1.5,-1.5) {};
			\node[A1,label=90:{$\ell_1$}] at (-1.5,-1.5) {};
			\draw (h) -- (v1) -- (v2);
			
			\node[draw,circle] (v3) at (1.5,0) {};
			\node[A1] at (1.5,0) {};
			\node[draw,circle] (v4) at (3,0) {};
			\node[A1] at (3,0) {};
			\node[draw,circle] (v5) at (4.5,0) {};
			\node[A1] at (4.5,0) {};
			\node[draw,circle] (v9) at (6,0) {};
			\node[A1,label=90:{$\ell_2$}] at (6,0) {};
			\node[draw,circle] (v6) at (1.5,-1.5) {};
			\node[A2] at (1.5,-1.5) {};
			\node[draw,circle] (v7) at (3,-1.5) {};
			\node[A2] at (3,-1.5) {};
			\node[draw,circle] (v8) at (4.5,-1.5) {};
			\node[A2] at (4.5,-1.5) {};
			\node[draw,circle] (v10) at (6,-1.5) {};
			\node[A2,label=90:{$\ell_3$}] at (6,-1.5) {};
			\draw (h) -- (v3) -- (v4) -- (v5) -- (v9);
			\draw (v4) -- (v6) -- (v7) -- (v8) -- (v10);
		\end{tikzpicture}
		\caption{An instance used to prove \Cref{thm:EF1:standardAlgorithmsDoNotWorkForNonWastefulness}. Highlighted is an EF1 and non-wasteful allocation.}
		\label{fig:EF1:standardAlgorithmsDoNotWorkForNonWastefulnessInstance}
	\end{figure}
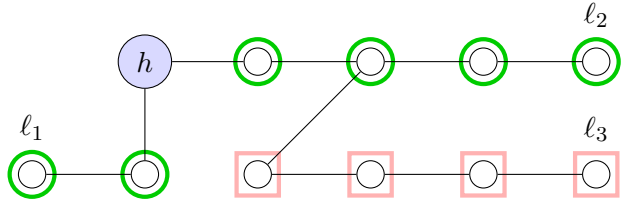
	
	\textbf{Round-robin mechanism.} The round-robin mechanism first fixes an ordering of agents, and then, the agents, one after another in the fixed order, select the order that increases the cost of their current bundle the least. The procedure stops once all orders are selected. As stated before, in a non-wasteful allocation, we are interested only in the allocation of leaves. Hence, our agents will select only between the leaves instead of selecting between all orders.
	
	Without loss of generality, assume that the mechanism assumes the agents in the lexicographic order. In the first round, agent~$1$ selects leaf~$\ell_1$, as it increases the cost of agent's~$1$ bundle only by two. Next, agent~$2$ chooses the leaf~$\ell_2$. At this point, the costs are as follows:~$\cost(\alloc_1) = 2$ and~$\cost(\alloc_2) = 4$. As there is only one remaining leaf~$\ell_3$ and its agent's~$1$ turn, the leaf~$\ell_3$ is added to~$\alloc_1$, and the mechanism stops. Regardless of how we complete the allocation~$\alloc$ while keeping it non-wasteful, the costs are~$\cost(\alloc_1) = \dist(h,\ell_1) + \dist(h,\ell_3) = 8$ and~$\cost(\alloc_2) = 4$. Clearly, the agent~$1$ has an envy towards agent~$2$; however, removing a single order from the agent's~$1$ bundle will decrease its cost by at most one, so the allocation is not EF1.
	
	\textbf{Envy-cycle elimination.} The envy-graph mechanism also works in rounds. Apart from the set of agents and the set of items, the mechanism also keeps track of a directed envy-graph~$D$, where the nodes (we use the term nodes instead of vertices to distinguish between vertices of~$G$ and vertices of the envy-graph~$G$) are agents, and there is an arc from agent~$i$ to agent~$j$ if agent~$i$ has an envy towards agent~$j$. Then, the algorithm fixes an ordering of the items, and in every round, it assigns the currently processed item to an agent that is not envious of any other agent (in general, it may happen that no such agent exists; however, in our example, this situation is not possible, so we do not discuss how to proceed in such a case).
	
	Let the ordering that the mechanism fixes be~$(\ell_1,\ell_2,\ell_3)$. Note that such an ordering is very natural, as it assigns the leaves based on their distance from the hub~$h$. In the first round, the envy graph is empty, so the first leaf is allocated, without loss of generality, to agent~$1$. This creates an envy from agent~$1$ towards agent~$2$, so the next item is allocated to agent~$2$. This reverses the direction of the single arc in~$D$, and in the final round, the leaf~$\ell_3$ is allocated to agent~$1$. Hence, we ended up in the same situation as with the round-robin mechanism, so the proof is complete.
\end{proof}

To conclude the section, we give at least one positive result. Specifically, we show that if the topology is a star, we can decide the existence of EF allocations in polynomial time. Moreover, in this case, EF1 allocations are guaranteed to exist and can be found efficiently.

\begin{proposition}\label{thm:envy:P:unweightedStar}
	If~$G$ is a star and the input instance is unweighted, an EF1 or EF and non-wasteful allocation can be found in polynomial time whenever one exists.
\end{proposition}
\begin{proof}
	Recall the allocation created in the proof of \Cref{thm:MMS:P:unweightedStar}. For EF1, it is easy to see that for this allocation~$\alloc$, it holds that~$\cost(\alloc_i) \in \{\lfloor \numOrders/\numAgents \rfloor, \lceil \numOrders/\numAgents \rceil \}$. Therefore, the difference between the cost of each pair of two bundles is at most one, and consequently, such a partition is clearly EF1.
	
	Next, let the desired fairness notion be envy-freeness (EF). We claim that there exists an envy-free allocation if and only if~$|\leaves(G)\setminus \{h\}| \bmod \numAgents = 0$. The left-to-right direction is clear; we create an allocation~$\alloc$ as in the case of MMS. By the construction, the cost of each bundle is the same, so the allocation is obviously EF. In the opposite direction, let~$\alloc$ be an EF allocation. By \Cref{lem:MMS:ifHubIsLeafWeCanReduce}, we can suppose that~$h$ is the center of~$G$. Therefore, all orders are leaves. We have already shown that the cost of each bundle~$\alloc_i$ is equal to~$|\alloc_i|$. Since~$\alloc$ is EF, the costs are identical,  and all orders are leaves, we directly obtain that~$|L(G)\setminus\{h\}|\bmod \numAgents = 0$, finishing the proof.
\end{proof}

\section{Concluding Remarks}
Our work extends the fair delivery problem to settings with weighted edges, proposes non-wastefulness as an efficiency concept, and provides a comprehensive landscape on designing tractable algorithms. 
The fixed-parameter and polynomial-time algorithms for computing MMS and non-wasteful allocations may give insights on further strengthening the efficiency notions to PO or other desirable concepts. 
Moreover, going beyond tree structures requires dealing with cycles, walks, which require new ways of modeling cost. It is not clear how the standard fairness notions, e.g. MMS, can be defined in this setting. While our negative computational results carry over to general graphs, tractable algorithms may arise when restricting the parameter/structure of the problem.

\bibliographystyle{plainnat} 
\bibliography{references}

\end{document}